\documentclass[oribibl,envcountsame,final]{llncs}
\usepackage{etex}

\usepackage{microtype}

\newcommand{\bbC}{\mathcal{C}_0\mathcal{M}} 
\newcommand{\CCM}{C_0M}
\newcommand{\inv}{^{-1}}
\newcommand{\DCPO}{\mathbf{DCPO}}

\newcommand{\To}{\Rightarrow}
\newcommand{\Pow}{\mathcal{P}}
\newcommand{\Nat}{\mathbb{N}}

\newcommand{\Dist}{\mathcal{D}} 

\newcommand{\Ke}{\mathop{\mathsf{Ker}}}
\newcounter{blubber}
\newenvironment{eenumerate}
{\begin{list}
  {(\roman{blubber})}
  {\usecounter{blubber}
   \setlength{\leftmargin}{0pt}
    \setlength{\parsep}{0pt}
    \setlength{\itemindent}{4ex}
    \setlength{\itemsep}{2pt}
  }
}
{\end{list}}

\usepackage[all]{xy}
\usepackage{proof}					
\usepackage[nosumlimits]{amsmath}
\usepackage{amssymb}				
\usepackage{stmaryrd,latexsym}		
\usepackage{wasysym}				
\usepackage{bbm}					
\usepackage{slashed}				

\usepackage{mathtools}

\usepackage{textcomp}

\usepackage[bbgreekl]{mathbbol}

\usepackage{tikz}
\usetikzlibrary{matrix}

\usepackage[colorlinks,linkcolor={blue},citecolor={blue},urlcolor={red},breaklinks=true,final]{hyperref}

\usepackage{hypcap}
\usepackage{ifdraft}  				

\makeatletter
\newcommand{\SG}[2][]{\ed@note{#2}{SG}{#1}}
\newcommand{\LS}[2][]{\ed@note{#2}{LS}{#1}}
\newcommand{\DP}[2][]{\ed@note{#2}{DP}{#1}}
\makeatother
\makeatletter
\newcommand*{\@old@slash}{}\let\@old@slash\slash
\def\slash{\relax\ifmmode\delimiter"502F30E\mathopen{}\else\@old@slash\fi}
\makeatother

\newcommand{\BBT}{\mathbb{T}}

\newcommand{\BBS}{\mathbb{S}}

\newcommand{\BBN}{\mathbb{N}}

\newcommand{\Set}{\Cat{Set}}

\newcommand{\Cat}{\mathbf}

\newcommand{\Hom}{\mathsf{Hom}}

\newcommand{\id}{\operatorname{id}}

\newcommand{\PSet}{{\mathcal P}}
\newcommand{\PFin}{{\mathcal P}_{\omega}}

\newcommand{\impl}{\Rightarrow}

\newcommand{\eps}{\operatorname\epsilon}
\newcommand{\comp}{\mathbin{\circ}}

\newcommand{\tpair}[2]{{#1\wideparen{\;}\kern-1pt #2}}
\newcommand{\argument}{{\operatorname{-\!-\!-}}}

\newcommand{\CASE}{\operatorname{\sf case}}
\newcommand{\OF}{\operatorname{\sf of}}
\renewcommand{\case}[3]{\CASE\kern1.2pt #1\kern1.2pt \OF\kern1.2pt #2;\kern1.2pt #3}

\newcommand{\caseOne}[2]{\CASE\kern1.2pt #1\kern1.2pt \OF\kern1.2pt #2}

\newcommand{\DO}{\operatorname{\sf do}}

\newcommand{\letTerm}[2]{\DO\kern1.2pt#1; #2}
\newcommand{\leteq}{\gets}

\newcommand{\letTermO}[2]{\DO_{\nu}\kern1.2pt#1; #2}
\newcommand{\LOOP}{\operatorname{\sf loop}}
\newcommand{\THEN}{\operatorname{\sf then}}

\makeatletter
\long\def\loopTerm@[#1][#2][#3]#4#5#6{
\DO\kern1.2pt#4 #1 \LOOP #3 #5 #2\THEN #6
}

\newcommand{\loopTerm}{
\optparams{\loopTerm@}{[;][\};][\{]}
}
\makeatother

\newcommand{\LET}{\operatorname{\sf init}}
\newcommand{\letLoop}[2]{\LET #1\kern1.2pt\LOOP\kern1.2pt\{#2\}}
\newcommand{\LetLoop}[2]{\LET #1\kern1.2pt\LOOP\kern1.2pt\bigl\{#2\bigr\}}

\newcommand{\IF}{\operatorname{\sf if}}
\newcommand{\ifTerm}[3]{\IF #1\kern2.2pt {\sf then}\kern1.2pt #2\kern2.2pt {\sf else}\kern2.2pt #3}
\newcommand{\ifTermO}[3]{\IF_{\nu} #1\kern2.2pt {\sf then}\kern2.2pt #2\kern2.2pt {\sf else}\kern2.2pt #3}

\newcommand{\WHILE}{\operatorname{\sf while}}
\newcommand{\whileTerm}[3]{\LET #1\kern1.2pt\WHILE\kern1.2pt #2 \kern2.2pt{\sf do}\kern2.2pt #3}
\newcommand{\whileTermS}[2]{\WHILE\kern1.2pt #1 \kern2.2pt{\sf do}\kern2.2pt #2}

\newcommand{\SEQ}{\operatorname{\sf seq}}
\newcommand{\seqTerm}[2]{\SEQ_{\nu}\kern1.2pt#1; #2}

\newcommand{\letin}[2]{\operatorname{\sf let} #1\kern2.2pt \operatorname{\sf in}\kern1.2pt #2}
\newcommand{\match}[2]{\operatorname{\sf let}\kern2.2pt #1\kern2.2pt \operatorname{\sf in}\kern1.2pt #2}
\newcommand{\letmon}[3]{\operatorname{\sf let}\kern1.2pt  #1\kern2.2pt \operatorname{\sf by} #2\kern2.2pt \operatorname{\sf in}\kern1.2pt #3}

\newcommand{\lettens}[4]{\operatorname{\sf let}\kern1.2pt  #3\kern2.2pt \operatorname{\sf by}  #1 (#2) \kern2.2pt \operatorname{\sf in}\kern1.2pt #4}

\newcommand{\lsem}{\llbracket}
\newcommand{\rsem}{\rrbracket}
\newcommand{\sem}[1]{\lsem #1 \rsem}

\newcommand{\brks}[1]{\langle #1\rangle}

\newlength{\myboxwidth}
\DeclareMathOperator{\redpar}{{
\declareslashed{}{
\vrule height2pt depth 2pt
\kern1pt
\vrule height2pt depth 2pt
}{0}{0}{\rightarrowtail}\slashed{\rightarrowtail}}}

\newcommand{\klstar}{\dagger}  					
\newcommand{\klcomp}{\mathbin{\diamond}}  		

\renewcommand{\emptyset}{\mathop{\slashed0}} 

\newcommand{\join}{\operatorname{\sqcup}}

\newcommand{\lub}{\operatorname{\bigsqcup}}

\spnewtheorem{thm}[theorem]{Theorem}{\bfseries}{\itshape}
\spnewtheorem{cor}[theorem]{Corollary}{\bfseries}{\itshape}
\spnewtheorem{cnj}[theorem]{Conjecture}{\bfseries}{\itshape}
\spnewtheorem{lem}[theorem]{Lemma}{\bfseries}{\itshape}
\spnewtheorem{lemdefn}[theorem]{Lemma and Definition}{\bfseries}{\itshape}
\spnewtheorem{prop}[theorem]{Proposition}{\bfseries}{\itshape}
\spnewtheorem{defn}[theorem]{Definition}{\bfseries}{\upshape}
\spnewtheorem{rem}[theorem]{Remark}{\bfseries}{\upshape}
\spnewtheorem{expl}[theorem]{Example}{\bfseries}{\upshape}
\spnewtheorem{thmdefn}[theorem]{Theorem and Definition}{\bfseries}{\itshape}
\spnewtheorem{propdefn}[theorem]{Proposition and Definition}{\bfseries}{\itshape}
\spnewtheorem{assumption}[theorem]{Assumption}{\bfseries}{\upshape}
\spnewtheorem{algorithm}[theorem]{Algorithm}{\bfseries}{\upshape}

\ifdraft{
 \usepackage[show]{ed} 
 \usepackage{todos}
 \pagestyle{plain}
}{
 \usepackage[hide]{ed}
 \renewcommand{\todo}[1]{}
 \pagestyle{plain}
}

\DeclareFontFamily{OT1}{pzc}{}
\DeclareFontShape{OT1}{pzc}{m}{it}{<-> s * [1.200] pzcmi7t}{}

\DeclareMathAlphabet{\mathpzc}{OT1}{pzc}{m}{it}

\newcommand{\pacman}[1]{}

\usepackage{needspace}

\newcommand{\nact}{\text{\textoneoldstyle}}
\newcommand{\zact}{\text{\textzerooldstyle}}

\newcommand{\longsquiggly}{\xymatrix{{}\ar@{~>}[r]&{}}}

\usepackage{enumitem}
\setlist{leftmargin=0cm,itemindent=.6cm,topsep=.2cm} 

\title{Coalgebraic Weak Bisimulation from Recursive Equations over Monads}
\author{Sergey Goncharov\inst{1} and  Dirk Pattinson\inst{2}}
\institute{Department of Computer Science, FAU Erlangen-N\"urnberg\and
Research School of Computer Science, Australian National University}

\begin{document}
\maketitle

\begin{abstract}
Strong bisimulation for labelled transition systems is one of the
most fundamental equivalences in process algebra, and has been
generalised to numerous classes of systems that exhibit richer
transition behaviour. Nearly all of the ensuing notions are
instances of the more general notion of \emph{coalgebraic bisimulation}.
Weak bisimulation, however, has so far been much less amenable to a
coalgebraic treatment. Here we attempt to close this gap by giving a coalgebraic
treatment of (parametrized) weak equivalences, including weak bisimulation. 
Our analysis requires that the functor defining the transition type of the system
is based on a suitable order-enriched monad, which allows us to capture weak
equivalences by least fixpoints of recursive equations.  Our notion is in agreement
with existing notions of weak 
bisimulations for labelled transition systems, probabilistic and
weighted systems, and simple Segala systems. 
\end{abstract}
\section{Introduction}
Both strong and weak bisimulations are fundamental
equivalences in process algebra~\cite{Milner89}. Both 
have been adapted to systems with richer behaviour 
such as probabilistic and weighted transition systems. 
For each class of systems, strong bisimulation is defined in a
similar way which is explained by universal coalgebra
where strong bisimulation is recovered as a canonical equivalence
that parametrically depends on the type of
system~\cite{Rutten:2000:UCT}. Weak bisimulations are much more difficult to analyse even for labelled transition
systems (LTS),
and much less canonical in status (e.g.~branching and delay bisimulations~\cite{GlabbeekWeijland96}).

We present a unified, coalgebraic treatment of various types of weak
bisimulation. 
An important special (and motivating) case of our definition is
\emph{probabilistic weak bisimulation} of 
Baier and Herrmanns~\cite{BaierHermanns97}. Unlike labelled
transition systems, probabilistic weak bisimulation needs to 
account for \emph{point-to-set} transitions, while \emph{point-to-point}
transitions, as for labelled transition systems, do not suffice:
Every LTS with a transition relation $\to$ induces an LTS with a \emph{weak transition relation} $\To$ and 
weak bisimulation for the original system is strong bisimulation of the
transformed one. This approach fails in the probabilistic case, as
weak point-to-point transitions no longer form a probability
distribution: in a system where $x \xrightarrow{a(0.5)} y$ and $x
\xrightarrow{\tau(0.5)} x$, we obtain $x
\xRightarrow{a(1)} y$ as the probability that $x$
evolves to $y$ along a trace of the form $\tau^*\cdot a\cdot \tau^*$
is clearly one, but  also $x \xRightarrow{\tau(1)} x$ as
the system will also evolve from $x$ to $x$ along $\tau^*$ also with
probability one (by simply doing nothing). Crucially, both events
are not independent.  This is resolved by
relating states to state sets along transition sequences, and 
the probability
$P(x,\Lambda,S)$ of $x$ evolving to a state in $S$ along a trace in
$\Lambda$ 
is the probability of the event that contains all execution
sequences leading from $x$ to $S$ via $\Lambda$, called \emph{total
probability} in \emph{op.cit.}
By re-formulating this idea axiomatically, we show that it is
applicable to a large class of systems, specifically coalgebras of
the form $X \to T(X \times A)$ where $T$ is enriched over 
directed complete partial orders with least element (pointed dcpos)
and non-strict maps.  Not surprisingly, similar (but stronger) assumptions
also play a prominent role in coalgebraic trace semantics
\cite{HasuoJacobsEtAl07}, and have two ramifications: the fact that
the functor $T$ that describes the branching behaviour extends to a
monad allows us to consider \emph{transition sequences}, and
order-enrichment permits us to compute the cumulative effect of
(sets of) transition sequences recursively using Kleene's fixpoint theorem. 
Our construction is parametric in an \emph{observation
pattern} that can be varied to obtain e.g.\ weak and delay
bisimulation. 
We demonstrate by example that our definition generalises concrete
definitions of probabilistic and weak weighted and probabilistic bisimulation found in the literature~\cite{BaierHermanns97,Corradini:1999:GMR,Segala:1994:PSP,Segala:1995:MVR}.  

A special role in our model is played by the operation of binary join, which is a continuous operation of the monad. We show that if it is also \emph{algebraic} in the sense of Plotkin and Power~\cite{PlotkinPower02}, which holds in the case of LTS, then weak bisimulation can be recovered as a strong bisimulation for a system of the same type, thus reestablishing Milner's weak transition construction. In the probabilistic case, for which join is unsurprisingly nonalgebraic, we show that weak bisimulation arises as strong bisimulation of a system based on the continuation monad.

\section{Preliminaries} \label{sec:prelims}

We use basic notions of category theory and coalgebra, see e.g.~\cite{Rutten:2000:UCT} for an overview. For a functor $F: \Set \to
\Set$, an \emph{$F$-coalgebra} is a pair $(X, f)$ with $f: X \to
TX$. Coalgebras form a category where the morphisms between $(X, f)$
and $(Y, g)$ are functions $\phi: X \to Y$ with $g \circ \phi = F \phi
\circ f$. A relation $E \subseteq X \times X$ is a \emph{kernel bisimulation} on $(X,
f)$ if there is an $F$-coalgebra $(Z, h)$ and two
morphisms $\phi: (X, f) \to (Z, h)$ and $\psi: (X, g) \to (Z, h)$
such that  $E = \Ke (\phi) = \lbrace (x, y) \in X \times X \mid \phi(x) =
\psi(y)\rbrace$ is the kernel of $\phi$. Clearly, kernel bisimulations are
equivalence relations, and we only consider kernel bisimulations in
what follows.  Kernel bisimulation agrees
with Aczel-Mendler bisimulation (and its variants) in case $F$ preserves pullbacks
weakly but is mathematically better behaved in case $F$ does not. It
also agrees (in all cases) with the notion of behavioural
equivalence: a thorough comparison is provided in 
~\cite{Staton:2011:RCN}.

We take \emph{monads} (on sets) as given by their extension form, i.e.\ as Kleisli
triples $\BBT = (T, \eta, \argument^\dagger)$ where $T: \Set \to \Set$ is a
functor, $\eta_X: X \to TX$ is a map for all sets $X$ and
$f^\dagger \colon TX \to TY$ is a map for all $f\colon X
\to TY$ subject to the equations $f^\dagger \eta_X =
f$, $\eta_X^\dagger = \id_{TX}$ and $(f^\dagger 
g)^\dagger = f^\dagger g^\dagger$ for all sets $X$
and all $f, g$ of appropriate type. Throughout, we write
$T$ for the underlying functor of a monad $\BBT$. The \emph{Kleisli
category} induced by a monad $\BBT$ has sets as
objects, but Kleisli-morphisms between $X$ and $Y$ are functions $f:
X \to TY$ with Kleisli composition $g \circ f = g^\dagger \circ f$
where $g^\dagger \circ f$ is function composition in $\Set$ and
$\eta_X$ is the identity at $X$. We 
use Haskell-style $\mathsf{do}$-notation to manipulate monad terms: for any $p\in TX$ and $q: X \to TY$ we write $\letTerm{x\leteq p}{q(x)}$ to denote $q^{\klstar}(p) \in TY$; if $p \in  T(X\times Y)$ we write $\letTerm{\brks{x,y}\leteq p}{q(x,y)}$.

In the sequel, we consider (among other examples) monads induced by
semirings: A \emph{semiring} is a structure $(R,
+, \cdot, 0, 1)$ such that $(R, +, 0)$ is a commutative monoid, $(R,
\cdot, 1)$ is a monoid and multiplication distributes over addition,
i.e.\ $x \cdot (y + z) = x\cdot y + x \cdot z$ and $(y+z) \cdot x =
y\cdot x + z \cdot x$.
A \emph{positively ordered semiring} is a semiring $(R, +, \cdot, 0,
1, \leq)$
equipped with a partial order $\leq$ that is positive ($0 \leq r$
for all $r
\in R)$ and compatible with the ring structure
($x\leq y$  implies that $x \Box z \leq y \Box z$ and
$z \Box x \leq  z \Box y$ for all $x, y, z \in R$ and $\Box
\in \lbrace +, \cdot \rbrace$). 
A \emph{continuous semiring} is a positively ordered semiring where
every directed set
$D \subseteq R$ has a least upper bound $\sup D \in R$ that is
compatible with the ring structure ($r \Box \sup D = \sup \lbrace
r
\Box  d \mid d \in D \rbrace$ and $\sup D \Box r = \sup \lbrace d
\Box r \mid d \in D \rbrace$ for all directed sets $D \subseteq
R$, all  $r \in R)$ and $\Box \in \lbrace +, \cdot \rbrace$.
Every continuous
semiring $R$ is a \emph{complete semiring}, i.e.\ has infinite sums given
by
$\sum_{i \in I} r_i = \sup \lbrace \sum_{i \in J}
r_i \mid J \subseteq I \mbox{ finite} \rbrace$.  We refer to
\cite{Droste:2009:SFP} for details.
If $R$ is complete, the functor $T_R X = X \to R$ extends to a
monad $\BBT_R$, called the~\emph{complete semimodule monad} (c.f.~\cite{Jacobs:2008:CTS}) with $\eta_X(x)(y) = 1$ if $x = y$ and $\eta_X(x)(y) = 0$,
otherwise, and $f^\dagger(\phi) (y) = \sum_{x \in X} \phi(x)
\cdot f(x)(y)$ for $f: X \to T_R Y$. Note if $R$ is continuous then all $T_R X$
are pointed dcpos under the pointwise ordering
of $R$ and the same applies to Kleisli homsets, i.e.\ the set of
Kleisli-maps of type $X \to TY$.

\section{Examples} \label{sec:examples}

We illustrate our generic approach to weak bisimulation by means of
the following examples. For all examples, strong bisimulation is
well understood  and known to coincide with kernel bisimulation.  As we will see later, the same is true for
weak bisimulation, introduced in the  next section.

\smallskip\noindent\textbf{Labelled Transition Systems.} We consider
the monad $\BBT_Q$ where $Q = \lbrace 0, 1 \rbrace$ is the boolean
semiring. Clearly $T_Q \cong \Pow$ where $\Pow$ is the covariant
powerset functor. A \emph{labelled transition system} can now be
described as a coalgebra $(X, f: X \to T_Q (X \times  A))$. It is well
known that bisimulation equivalences on labelled transition systems
coincide with kernel bisimulations as introduced in the previous
section.

\smallskip\noindent\textbf{Probabilistic Systems.} Consider the
monad $\BBT_{[0, \infty]}$ induced by the complete semiring of
non-negative real numbers, extended with infinity. Various types of
probabilistic systems arise as sub-classes of systems of type $(X,
f: X \to T_{[0, \infty]}(X \times A))$. For \emph{reactive
systems}, one postulates $\sum_{y \in X} f(x)(y, a) \in
\lbrace 0, 1 \rbrace$ for all $x \in X$ and all $a \in A$. 
\emph{Generative systems} satisfy $\sum_{(y, a) \in Y
\times A} f(x)(y, a) \in \lbrace 0, 1 \rbrace$ for all $x \in X$,
and \emph{fully probabilistic systems} satisfy $\sum_{(y, a) \in X
\times A} f(x)(y, a) = 1$ for all $x \in X$.
We refer to~\cite{Bartels:2003:HPS} for a detailed analysis of
various types of probabilistic systems in coalgebraic terms. It is
known that probabilistic bisimulation equivalence
\cite{LarsenSkou91} and kernel bisimulations agree
\cite{VinkRutten99}.  Our justification of viewing these various
types of probabilistic systems as $[0, \infty]$ weighted transition
systems comes from the fact that kernel bisimulations are
reflected by embeddings:

\begin{lem} \label{lem:emb}
Let $\kappa: F \to G$ be a monic natural transformation between two
set-functors $F$ and $G$ and $(X, f)$ be an $F$-coalgebra. Then
kernel bisimulations on the $F$-coalgebra $(X, f)$ agree with kernel
bisimulations on the $G$-coalgebra $(X, \kappa_X \circ f)$.
\end{lem}
\textbf{Integer Weighted Transition Systems.}  Weighted
transition systems, much like probabilistic systems, arise as
coalgebras for the functor $FX = T_{\Nat \cup \lbrace \infty
\rbrace} (X \times A)$ where $\Nat \cup \lbrace \infty \rbrace$ is
the (complete) semiring of natural numbers extended with $\infty$
and the usual arithmetic operations.  In an (integer) weighted
transition system, every labelled transition comes with a weight,
and we can write $x \stackrel{a(n)}{\longrightarrow} y$ if $f(x)(y,
a) = n$. In process algebra, weights represent different ways in
which the same transition can be derived syntactically, e.g.\ 
$a.0 + a.0 \xrightarrow{a(2)} 0$, according to the reduction of the
term on the left and right, respectively. The ensuing (strong)
notion of equivalence has been studied in 
\cite{Aceto:2010:RBG} and shown to be coalgebraic. 

The three examples above are a special instance of semiring-weighted
transition systems, studied for instance in
\cite{Latella:2012:BLS}. This is not the case for systems that
combine probability and non-determinism.

\smallskip\noindent\textbf{Non-Deterministic Probabilistic Systems.} 
As we have motivated in the introduction, a coalgebraic analysis of
weak bisimulation hinges on the ability to sequence transitions,
i.e.\ the fact that the functor $F$ defining the concrete shape of a
transition system $(X, f: X \to F(X \times A))$ extends to a monad.
The naive combination of probability and non-determinism, i.e.\
considering the functor $F = \Pow \circ \Dist$ where $\Dist(X)$ is
the set of finitely distributed probability distributions 
does not extend to a monad~\cite{VaraccaWinskel06}. One solution,
discussed in \emph{op.cit.}  and elaborated in
\cite{Jacobs:2008:CTS} is to restrict to \emph{convex} sets of
valuations. 
Informally, we use monad $\bbC$ (a variant of the $\mathcal{CM}$ monad from~\cite{Jacobs:2008:CTS}),
encompassing two semiring structures, for probability and
non-determinism,
and the former
distributes over the latter, i.e.\ $a +_p  (b+c) = (a +_p b) +
(a +_p c)$ 
where $+$ is nondeterministic choice and $+_{p}$ is
probabilistic choice (choose `left' with probability $p$ and 'right'
with probability $1-p$).
Concretely, for the underlying functor $\CCM$ of the monad $\bbC$,
$\CCM X$ is the set of nonempty
\emph{convex} sets of finite valuations over $[0,\infty)$, i.e.\ finitely supported maps to $[0,\infty)$, containing
the trivial valuation identically equal to~$0$. 
A set $S$ is convex
if
$\sum_i r_i\cdot\xi_i \in S$
whenever all $\xi_i\in S$ and $\sum_i r_i=1$. Our
definition deviates slightly from~\cite{Jacobs:2008:CTS} in
that we require that $\CCM X$ contains the zero valuation, whereas
in~\emph{op.cit.} (and also in~\cite{Brengos13}) this condition is used to restrict the class of
systems to which the theory is applied. 

\section{Weak Bisimulation, Coalgebraically}

Capturing weak bisimulation for transition systems $(X, f: X
\to T(X \times A))$ coalgebraically, where $A$ is a set of labels
that we keep fixed throughout,  amounts to two
requirements: first, $T$ needs to extend to a monad which enables us
to sequence transitions. Second, we need to be able to compute the
cumulative effect of transitions which requires the
monad to be enriched over the category of directed complete partial
orders (and non-strict morphisms).  

\begin{defn}[Completely ordered monads]\label{defn:com}
A monad $\BBT$ is \emph{completely ordered} if its Kleisli
category is enriched over the category $\DCPO_{\bot}$ of
directed-complete partial orders with least  element (pointed dcpos) 
and continuous
maps:
every hom-set $\Set(X, TY)$ is a pointed dcpo
and Kleisli composition is continuous, i.e.\ the joins
%
$
f^\klstar\comp\left(\bigsqcup_i g_i\right) =\bigsqcup_i f^\klstar\comp g_i$ and
$\left(\bigsqcup_i f_i\right)^\klstar\comp g =\bigsqcup_i f_i^\klstar\comp g $
exist and are equal whenever the  join on the left hand side is taken over a directed
set. 
A \emph{continuous operation} of arity $n$ on a completely ordered
monad is a natural transformation $\alpha: T^n \to T$ for which
every component $\alpha_X$ is Scott-continuous.
\end{defn}
The diligent reader will have noticed that the same type of
enrichment is also required in the coalgebraic treatment of trace
semantics \cite{HasuoJacobsEtAl07}. This is by no means a surprise,
as the observable effect of weak transitions are precisely given in
terms of (sets of) traces. 

Often, these sets are defined in terms of weak transitions of the
form $\stackrel{\tau}{\to}^* \cdot \stackrel{a}\to \cdot
\stackrel{\tau}{\to}^*$. We think of weak
transitions as transitions along trace sets closed under Brzozowski
derivatives which enables us to recursively decompose a weak transition into a
(standard) transition, followed by a weak transition. 

\begin{defn}[Observation pattern]\label{defn:obs} An
\emph{observation pattern} over a set $A$ of labels is a subset $B \subseteq \Pow(A^*)$ that is
closed under Brzozowski derivatives, i.e.\ $b/a = \lbrace w \in
A^* \mid aw \in b \rbrace \in B$ for all $b \in B$ and all $a \in
A$.
\end{defn}
\noindent
Different observation patterns capture different notions of weak
bisimulation:
\begin{expl}[Observation patterns]\label{expl:obs} 
Let $A$ contain a silent action $\tau$.
\begin{eenumerate}
\item the \emph{strong pattern} over $A$ is given by $B = \lbrace
\lbrace a \rbrace \mid a \in A \rbrace \cup \lbrace \emptyset, 
\lbrace \epsilon \rbrace \rbrace$. 
\item the \emph{weak
pattern} over $A$ is given by  $B = \lbrace \hat a \mid a \in A
\rbrace$ where $\hat \tau = \tau^*$ and $\hat a = \tau^* \cdot a
\cdot \tau^*$ for $a \neq \tau$.
\item The \emph{delay pattern} is $B = \lbrace \tau^* a
\mid a \in A \setminus \lbrace \tau \rbrace \rbrace \cup \lbrace
\tau^* \rbrace$.
\end{eenumerate}
It is immediate that all are closed under Brzozowski derivatives.
\end{expl}

%
%
%

\noindent
Given an observation pattern that determines the notion of traces,
our definition of weak bisimulation relies on the fact that the
cumulative effect of transitions can be computed recursively. This
is ensured by enrichment, and we have the following (see Section
\ref{sec:prelims} for the $\mathsf{do}$-notation):

\begin{lem} \label{lem:lfp} Suppose $B$ is an observation pattern over $A$, 
$\BBT$ is a completely ordered monad and $\oplus:
T^2 \to T$ is continuous. 
Then the equation
\begin{align}\label{eq:rec}\tag{$\bigstar$}
f_{h}^B(x)(b) = 
  \left.\begin{cases}
	\eta(h(x)) & \mbox{if } \epsilon \in b \\ \bot & \mbox{otherwise}
	\end{cases} \right\}
  \oplus \letTerm{\brks{y,a}\leteq f(x)}{f_{h}^B(y)(b/a)}
\end{align}
has a unique least solution $f_{h}^B:X\to (TY)^B$ 
for all $f: X \to
T(X \times A)$ and all $h: X \to Y$.
\end{lem}
\noindent
Lemma~\eqref{lem:lfp} follows from 
Kleene's fixpoint theorem~\cite{Winskel93}
using order-enrichment.
The central notion of our paper can now be given as follows:

\begin{defn}\label{defn:sbis} 
Suppose that $\BBT$ is a completely ordered monad with a continuous
operation $\oplus$, $B$ is
an observation pattern over $A$ and let $f: X \to
T(X \times A)$. 
An equivalence relation $E \subseteq X \times X$ is a
\emph{$B$-$\oplus$-bisimulation} if  $E \subseteq  \Ke(f^B_\pi)$
where $\pi: X \to X \slash E$ is the canonical projection
(and $f_{\pi}^B$ is the unique least solution
of~\eqref{eq:rec}). We often elide the continuous operation, and say that $x, x' \in X$ are
$B$-bisimilar, if they are related by a $B$-bisimulation.  
\end{defn}

\noindent
Some remarks are in order before we show that the above definition agrees with various notions
of weak bisimulation studied in the literature.

\begin{rem}
\begin{eenumerate}
\item Intuitively, the requirement $E \subseteq \Ke(f^B_\pi)$
expresses that any two $E$-related states $x$ and $x'$ have the same
cumulative behaviour under all trace sets in $B$, \emph{provided}
that $E$-related states are not distinguished. In other words, a
state $[x]_E$ of the quotient of the original system exhibits the
same behaviour with respect to all trace sets in $B$, as the
representative $x$ of $[x]_E$. This intuition is made precise in
Section \ref{sec:alg} where we show how $B$-bisimulation can be
recovered as strong bisimulation (and hence quotients can be
constructed). 
\item The definition of weak bisimulation above caters for systems
of the form $(X, f: X \to T(X \times A))$, i.e.~we implicitly
consider the labels as part of the observable behaviour, or as
'output'. The role of labels appears to be reversed when computing
the cumulative effect of transitions via the function $f^B_\pi:
X \to T(X \slash E)^B$. This apparent reversal  of roles is due to
the fact that every element of $B$ is a set of traces.
Accordingly, the function
application
$f^B_h(x)(b)$ represents the totality of behaviour that can be
observed along traces in $b$, starting from $x$, and trace sets are
now 'input'.
\end{eenumerate}
\end{rem}

\noindent
As a slogan, $B$-bisimilarity is a
$B$-bisimulation:
\begin{lem} \label{lem:union}
Let $(E_i)_{i \in I}$ be a family of $B$-bisimulation
equivalences on $(X, f: X \to T(X \times A)$. Then so is the
transitive closure of $\bigcup_{i \in I} E_i$.
\end{lem}

\section{Examples, Revisited}\label{eq:expl}
We demonstrate that $B$-bisimulation
agrees with the
known (and expected) notion of weak bisimulation for the examples
in Section \ref{sec:examples}. To instantiate the general definition to coalgebras of
the form $X \to T(X \times A)$, we need to verify that the monad
$\BBT$ is completely ordered. 
This is the case for complete semimodule monads over continuous semirings.

\begin{lem} \label{lem:semiring-enriched} Let $R$ be a continuous
semiring. Then the monad $\BBT_R$
is completely ordered, and both join $\sqcup$ and semiring sum $+$
are continuous operations on $T$.
\end{lem}

\noindent
This lemma in particular ensures that $B$-bisimulation is
meaningful for transitions systems weighted in a complete semiring,
and in particular for labelled, probabilistic and integer-weighted
systems.

\smallskip\noindent\textbf{Labelled transitions systems.}
As in Section \ref{sec:examples}, labelled transition systems are coalgebras for the
functor $FX = \Pow(X \times A)$. For an $F$-coalgebra $(X, f)$, 
Equation~\eqref{eq:rec} stipulates that
%
\[
  f_{h}(x)(b) = \bigl\{h(x)\bigm\vert \epsilon \in b \bigr\} \cup \bigcup_{x\xrightarrow{a} y} f_{h}(y)(b/a)
\]
where $x\xrightarrow{a} y$ iff $\brks{y,a}\in f(x)$. By
Kleene's fixpoint theorem,  the least solution is
\begin{align*}
f_{h}(x)(b) = \Bigl\{h(x_k)\bigm\vert x\xrightarrow{a_1} x_1\xrightarrow{a_2}\ldots\xrightarrow{a_k} x_k,~~a_1\cdots a_k\in b\Bigr\}.
\end{align*}
If $\tau \in A$, 
$B$ is the weak pattern and $E \subseteq X \times X$
is an equivalence, this gives
\[ [x']_E \in f^B_\pi(x)(\hat a) \mbox{~~~iff~~~} x \stackrel{\hat
a}{\Longrightarrow} x' \]
where $x \stackrel{\hat a}{\Longrightarrow} x'$ if there are $(y_1,
a_1), \dots, (y_n, a_n)$ such that $x \stackrel{a_1}\to y_1
\stackrel{a_2}\to  \dots \stackrel{a_n}\to y_n = x'$ and $a_1 \cdots a_n
\in \hat a$.
By Definition~\ref{defn:sbis}, $E$ is a $B$-bisimulation if for
any $\brks{x,y}\in E$, $\bigl\{[x']_E\mid x\stackrel{\hat a}\To
x'\bigr\}=\bigl\{[y']_E\mid y\stackrel{\hat a}\To y'\bigr\}$ for any
$a\in A$ (including $\tau$). The latter is easily shown to be
equivalent to the standard notion of weak bisimulation equivalence.
By analogous reasoning one readily recovers delay 
bisimulation equivalences from the delay pattern.

\smallskip\noindent\textbf{Probabilistic systems.} 
Fully probabilistic system (Section \ref{sec:examples})
are coalgebras of
type $(X, f: X \to T_{[0, \infty]}(X \times A))$, where $\BBT_{[0,
\infty]}$ is the complete semimodule monad induced by $[0, \infty]$
and additionally satisfy $\sum_{(y, a)
\in X \times A} f(x)(y, a) = 1$ for all $x \in X$. In
\cite{BaierHermanns97}, an equivalence relation $E \subseteq X \times X$ is
a weak bisimulation, if 
\[ P(x, \hat a, [y]_E) = P(x', \hat a, [y]_E) \]
for all $a \in A, y \in X$ and $(x, x') \in E$. Here $\hat a$ is
given as in Example \ref{expl:obs} and $P(x, \Lambda, C)$ is
the \emph{total probability} of the system evolving from state $x$
to a state in $C$ via a trace in $\Lambda \subseteq A^*$.
\emph{Op.cit.} states that total
probabilities satisfy the recursive equations: $P(x, \Lambda, C) = 1$
if $\epsilon \in C$ and $x \in \Lambda$, and
\[ 
  P(x, \Lambda, C) = \sum_{(y, a) \in X \times A} f(x)(y, a) \cdot
P(y, \Lambda / a, C) \]
otherwise.  In fact, total probabilities are the \emph{least}
solution (with respect to the pointwise order on $[0, \infty]$) of
the recursive equations above. 

\begin{lemma} \label{lemma:fully-probabilistic}

Let $(X, f: X \to T_{[0, \infty]} (X \times A))$ be a fully
probabilistic system, $B$ an observation pattern over $A$ and $E \subseteq X \times X$ an equivalence
relation. If $\pi: X \to X \slash E$ is the canonical projection,
then $P(x, b, [y]_E) = f^B_\pi(x)(b)([y]_E)$ for all $x, y \in
X$ and all $b \in B$, using $\sqcup$ as continuous operation.
\end{lemma}

\noindent
As a corollary, we obtain that weak bisimulation of fully
probabilistic systems is a special case of $B$-bisimulation for
the weak pattern.

\smallskip\noindent\textbf{Weighted transition systems.}
Weighted transition systems are technically similar to probabilistic
systems as they also appear as coalgebras for a (complete)
semimodule monad, 
but without any restriction on the sum of weights.
The associated notion of \emph{weak resource bisimulation} is
described syntactically in \cite{Corradini:1999:GMR}.  Abstracting
from the concrete syntax and taking weighted transition systems as
primitive, we are faced with a situation that is reminiscent of the
probabilistic case: a weak resource bisimulation equivalence on a
weighted transition system $(X, f: X \to T_{\Nat \cup \infty} (X
\times A))$ is an
equivalence relation $E \subseteq X \times X$ such that $xEy$ and $a
\in A$ implies that $W(x, \Lambda , C) = W(y, \Lambda , C)$ for all
equivalence classes $C \in X \slash E$ and all $\Lambda$ that are of
the form $\tau^*a\tau^*$ for $a\neq \tau$ and $\tau^*$. Here $W(x,
\Lambda, C)$ is the \emph{total weight}, i.e.\ the maximal number of
possibilities in which $x$ can evolve into a state in $C$ via a path
from $\Lambda$. Total weights can be understood as (weighted) sums
over all \emph{independent} paths that lead from $x$ into $C$ via a
trace in $\Lambda$, where two paths are independent if neither is a
prefix of the other. 
Analogously to the probabilistic case, these weights
are given by the least solution of the recursive equations
\[
W(x, \Lambda, C) =  \left. 
  \begin{cases} 1 & \epsilon \in \Lambda, x \in C \\
  0 & \mbox{otherwise} \end{cases}\right\} \sqcup
  \sum\limits_{(y, a) \in X \times A} f(x)(y, a) \cdot W(y, \Lambda /
a, C)
\]
and represent the total number of possibilities in which a process
$x$ can evolve into a process in $C$ along a trace in $\Lambda$. 
For example, we have that $W(0 + \tau.0 + \tau.\tau.0, \tau^*,
\lbrace 0 \rbrace )= 3$ representing the three different
possibilities in which the given process can become inert along a
$\tau$-trace, and $W(x, \tau^*, z) = 6$ for the triangle-shaped
system $x \stackrel{\tau(2)}\to y$, $x \stackrel{\tau(2)}\to z$ and $y
\stackrel{\tau(2)}\to  z$. It is routine to check that $W(x, b,
[x']_E) = f^B_\pi(x)(b)([x']_E)$. 
Unlike the probabilistic case,  
the number of different ways in which
processes may evolve is strictly additive.  For the weak pattern,
$B$-bisimulation is therefore the semantic manifestation of
weak resource bisimulation advocated in \cite{Corradini:1999:GMR}.
%
%
%
%

\medskip\noindent
\textbf{Probability and nondeterminism.}
Systems that combine probabilistic and nondeterministic behaviour arise
as coalgebras of type $(X, f: X \to \CCM(X \times A))$ where $\bbC$
is the monad from Section \ref{sec:examples}. Systems of this type capture so-called Segala systems. Here we stick to simple Segala systems, which are colagebras of type $\PSet(\Dist\times A)$ and for which the ensuing notion of weak probabilistic bisimulation was introduced in~\cite{Segala:1994:PSP}. 
These systems extend probabilistic
systems by additionally allowing non-deterministic transitions. 
As was essentially elaborated in~\cite{Brengos13}, every simple Segala system embeds into a coalgebra $(X, f: X \to \CCM(X \times A))$.

Completing a simple Segala system to a coalgebra over $\bbC$ amounts to forming \emph{convex} sets
of valuations; convexity arises from probabilistic choice as follows: given
non-deterministic transitions $x \to \xi$ and $x \to \zeta$, where
$\xi$ and $\zeta$ are valuations over $X \times A$ induces a
transition $x \to \xi +_p \zeta$ where $+_p$ is probabilistic choice.
Following~\cite{VaraccaWinskel06}, one way to understand this is to also consider non-deterministic choice
$+$ and to observe that 
\begin{displaymath}
\xi + \zeta = (\xi + \zeta + \xi) +_p (\xi + \zeta + \zeta) = (\xi + \zeta) +_p (\xi + \zeta) + (\xi +_p \zeta) = \xi + \zeta
+ (\xi +_p \zeta)
\end{displaymath} by the axioms $\xi + \xi = \xi +_p \xi = \xi$, $(\xi + \zeta) +_p \theta = (\xi +_p \theta) + (\zeta +_p \theta)$, the
last one describing the interaction between probabilistic and
non-deterministic choice. 

%
We argue that $B$-bisimulation where $B$ is the weak observation
pattern agrees with the notion
from~\cite{Segala:1994:PSP,Segala:1995:MVR}. We make a forward
reference to 
Theorem~\ref{thm:red}  which shows that $B$-bisimulation for $(X,f)$
amounts to strong bisimulation for $(X,f_{\id}^B)$. In other words,
weak bisimilarity can be recovered from strong bisimilarity for the
system whose transitions are weak transitions of the original
system. Solving the recursive equation for $f^B_{\id}$ (where $B$ is
the weak pattern and we use the notation of Example \ref{expl:obs})
we can write $x \stackrel{\hat a}\Longrightarrow \xi$ if $\xi \in
f^B_{\id}(x)(a)$. Intuitively, this represents that $x$ can evolve along
a trace in $\hat a$ to the valuation $\xi$, interleaving
probabilistic and nondeterministic steps. We then obtain that an
equivalence relation $E \subseteq X \times X$ is a $B$-bisimulation
if, whenever $(x, y) \in E$ and $x \stackrel{\hat a}\Longrightarrow
\xi$, there exists $\zeta$ such that $y \stackrel{\hat
a}{\Longrightarrow} \zeta$ and $\xi$ and $\zeta$ are 'equivalent up to
$E$', that is, $(F\pi) \xi = (F \pi) \zeta$ where $\pi: X \to X \slash E$
is the projection and $FX = [0, \infty)^X$. 
More concretely, the weak relation
$\mathop{\To} \in X \times B \times [0, \infty)^X$
is obtained by~\eqref{eq:rec} and is the 
least solution of the following system:
\begin{align*}
x&\xRightarrow{\hat\tau}\delta_{x}\\
x&\xRightarrow{\hat a}\zeta \text{~~~iff~~~} \exists \xi\in f(x).~\zeta\in\left\{\sum_{y\in X}\xi(y,a)\cdot\theta^\tau_{y}+\xi(y,\tau)\cdot\theta^a_y\mathop{\bigm\vert }\forall y.\, y\xRightarrow{\hat b}\theta^b_{y}\right\}  
%
\end{align*}
where $x\xRightarrow{\hat b}\zeta$ ($b\in\{a,\tau\}$) abbreviates
$\brks{x,b,\zeta}\in\mathop{\To}$; $\delta_{y}(y')=1$ if $y=y'$ and
$\delta_{y}(y')=0$ otherwise; and scalar multiplication and
summation act on valuations pointwise. Kleene's fixpoint theorem underlying Lemma~\ref{lem:lfp} ensures that the relation $\To$ can be calculated  iteratively, 
i.e.\
$\mathop{\To}=\bigcup_i\mathop{\To_i}$ where the $\To_i$ replace $\To$ in the above
recursive equations in the obvious way, hence making them recurrent.
Then $x\xRightarrow{\hat b}\zeta$ iff there is $i$
such that $x\xRightarrow{\hat b}_i\zeta$. 
The resulting definition in terms of weak transitions $\To_i$ matches weak probabilistic
bisimulation from~\cite{Segala:1994:PSP,Segala:1995:MVR}. 
%
Note that convexity of the monad precisely ensures that $\xi$ in the recursive clause above for $x\xRightarrow{\hat a}\zeta$ represents a
\emph{combined step} of the underlying Segala system, which by
definition, is exactly a convex combination of ordinary
probabilistic transitions.


\section{Weak Bisimulation as Strong Bisimulation}\label{sec:alg}
Milner's weak transition construction characterises weak
bisimilarity as bisimilarity for a (modified) system whose
transitions are the weak transitions of the original system. 
This construction does not 
transfer to the general case, witnessed by the case of (fully)
probabilistic systems. The pivotal role is played by the continuous
operation $\oplus$ that determines $B$-bisimulation.
We show that Milner's
construction generalises if $\oplus$ is \emph{algebraic} and present
a variation of the construction if algebraicity fails.
An \emph{algebraic operation} of arity $n$ on a monad $\BBT$ (e.g.~\cite{PlotkinPower02}) is a natural
transformation $\alpha: T^n\to T$ such that $\alpha_Y \circ (f^\klstar)^n = f^\klstar \circ
\alpha_X$ for all $f: X \to TY$. Algebraic operations are
automatically continuous:
\begin{lem}\label{lem:alg}
Algebraic operations of completely ordered monads are continuous.
\end{lem}
\begin{expl}[Algebraic operations]
Semiring summation $+$ is algebraic on continuous semimodule monads. 
If the underlying semiring is idempotent, e.g.\ the boolean
semiring, 
summation coincides with the join operation $\sqcup$  which is
therefore also algebraic. The bottom element
$\bot$ is a nullary algebraic operation (constant). 
The join operation is algebraic on the monad $\bbC$ from Section
\ref{sec:examples}.
%
%
The join operation $\join$ is generally not algebraic for free
(complete) semimodule monads unless the semiring is idempotent.
\end{expl}

\noindent
Algebraicity of $\oplus$ allows us to lift Milner's construction to
the coalgebraic case: 
$B$-bisimulations coincide with kernel
bisimulations for a modified system 
of the \emph{same transition type}. This instantiates to labelled
transition systems, as $\sqcup$ is algebraic on the semimodule monad
induced by the boolean semiring. 
We
show this using a sequence of lemmas, the first asserting that
algebraic operations commute over fixpoints. 
\begin{lem}\label{lem:rec_tail}
Suppose $h:X\to Y$ and $u:Y\to Z$. Given a coalgebra $(X, f: X \to
T(X \times A))$ we have that
$f^B_{u h}=T^Bu\comp
f^B_{h}$ if $\oplus$ is algebraic. 
\end{lem}
Similarly, sans algebraicity, 
$B$-bisimulations commute with morphisms.

\begin{lem}\label{lem:dia} Let $h:X\to Y$ be a morphism from $(X,f:X \to T(X \times A))$ to $(Y, g:
Y \to T(Y \times A))$. Then $g_u^B \circ
h = f^B_{u h}$ for all $u:  Y \to Z$. 
\end{lem}
Consequently, 
kernel bisimulations are $B$-bisimulations:
\begin{cor}
Let $h:X\to Y$ be a morphism of coalgebras $(X, f: X \to T(X \times
A))$ and $(Y, g: Y \to T(Y \times A))$. 
Then $\Ke h\subseteq\Ke
f^B_{h}$.
\end{cor}

\noindent
Lemma \ref{lem:rec_tail} shows that for monads equipped with an algebraic operation
$\oplus$ (such as the monad 
defining) labelled transition systems, we can recover
$B$-bisimilarity as strong bisimilarity of a transformed
system.
\begin{thm}\label{thm:red}
Provided $\oplus$ is algebraic, $E$ is a $B$-bisimulation on a monad-type coalgebra $(X, f)$ iff $E$ is a kernel
bisimulation equivalence on $(X, f^B_{\id})$. 
\end{thm}
If $\oplus$ is not algebraic it can still be possible to recover $B$-bisimulation as a kernel bisimulation for a system of a different type. For probabilistic systems this was done in~\cite{SokolovaVinkEtAl09}. Here, we obtain a similar result in a more conceptual way using the~\emph{continuous continuation monad} $\BBT$, which is obtained from the standard continuation monad~\cite{Moggi91} by restricting to continuous functions: the functorial part of $\BBT$ is $TX = (X \to D) \to_c D$ where $\to_c$ it the continuous function space, $D$ is a directed-complete partial order, and $(X \to D)$ is ordered pointwise. 
\begin{lem}\label{lem:conalg}
For a pointed dcpo $D$, $TX = (X \to D) \to_c D$ extends to a submonad $\BBT$ of the corresponding continuation monad, $\BBT$ is completely ordered, and every $\oplus:T^2\to T$, given pointwise, i.e.\ $(p\oplus q)(c)=p(c)\oplus q(c)$, is algebraic.
\end{lem} 

\noindent
The following lemma is the $B$-bisimulation analogue of Lemma
\ref{lem:emb} and is the main technical tool for reducing
$B$-bisimulation to kernel bisimulation.
\begin{lem}\label{lem:ext} Let $(X,f:X\to T(X \times A))$ be a 
coalgebra and $\kappa: T \to \widehat{T}$ an injective
monad morphism. 
If $\widehat\oplus$ is an algebraic operation on $\widehat{T}$
such that
$\widehat\oplus \mathop{\circ} \kappa^2 = \kappa\mathop{\circ}
\oplus$ then
$B$-$\oplus$-bisimulation equivalences on $(X, f)$ and
$B$-$\widehat\oplus$-bisimulation equivalences 
on $(X, \kappa f)$ agree.
\end{lem}
We use Lemma~\ref{lem:ext} as follows. 
Given a complete semimodule monad $\BBT$ over a (complete)  semiring
$R$, we embed $TX$ into $\widehat TX = (X \to T1) \to_c T1$ (where $T1 =
R$) by mapping $p \in TX$ to the function $\lambda c:X \to T1.
c^\dagger(p)$. 
This embedding is injective, and 
the conditions of Lemma~\ref{lem:ext} are fulfilled with $\oplus =
\sqcup$ and $\widehat\oplus$ the pointwise extension of $\oplus$ (which
is algebraic by Lemma \ref{lem:conalg}). 
 This
gives: 
\begin{thm} \label{cor:semi-strong}
Let $\BBT$ be a continuous semimodule monad over a continuous semiring $R$. Let $(X,f:X\to (X \times A))$
be a coalgebra and let $\oplus$ be the join on $R$.
Then $E$ is a $B$-bisimulation equivalence on $(X,f)$
iff it is a bisimulation equivalence on $(X,(\kappa_X \circ f)^B_{\id} : X \to (X \times A\to_c R) \to R)$.
\end{thm}
In summary, Milner's weak transition construction generalises to the
coalgebraic case if $\join$ is algebraic, and lifts to
a different transition type for semirings.

%

%

%
%
%


\section{Conclusions and Related Work}
We have presented a generic definition, and basic structural
properties, of weak bisimulation in a general, coalgebraic
framework. We use coalgebraic methods and 
enriched monads, similar to the coalgebraic treatment of trace
semantics \cite{HasuoJacobsEtAl07}.  Our definition applies
uniformity
to labelled transition systems,  probabilistic and weighted systems,
and to Segala systems from~\cite{Segala:1994:PSP}. 
Most of our results, including the notions of $B$-bisimulation 
as a solution of the recursive equation~\eqref{eq:rec}, easily transfer to categories other than $\Set$. 
An important conceptual contribution is the fact that algebraicity
allows to generalise Milner's weak transition construction to the
coalgebraic setting (Theorem \ref{thm:red}),  recovering
$B$-bisimulation as kernel bisimulation for a (modified) system of
the same transition type. We also provide an alternative for cases
where this fails (Theorem~\ref{cor:semi-strong}).

\smallskip\noindent\textbf{Related work.}
Results similar to ours are presented both in~\cite{Brengos13} and
~\cite{MiculanPeressotti13}. 
Brengos~\cite{Brengos13} uses a remarkably similar
tool set (order-enriched monads) but in a
substantially different way: Given a system of type
$T(F+\argument)$ with $\BBT$ order-enriched, 
the monad structure on $\BBT$ extends to 
$T(F+\argument)$, and saturation w.r.t.\ internal transitions is
achieved by iterating the obtained monad in a way resembling the
weak transition construction for LTS. Examples include labelled
transition systems and 
(simple) Segala systems. For both underlying monads, join is
algebraic, so that both examples are covered by our lifting Theorem~\ref{thm:red}. Fully probabilistic systems, for which algebraicity
fails, are not treated in~\cite{Brengos13}. 
Miculan and Peresotti~\cite{MiculanPeressotti13} also approach weak
bisimulation by solving recurrence relations, but only treat
(continuous) semimodule monads and do not account for (simple) Segala
systems. Our treatment covers all examples considered in both
~\cite{Brengos13} and ~\cite{MiculanPeressotti13}, and additionally
identifies the pivotal role of algebraicity in the generalisation of
Milner's construction. 
Sokolova et.al.~\cite{SokolovaVinkEtAl09} are concerned
with probabilistic systems only and reduce probabilistic weak
bisimulation to strong (kernel) bisimulation for 
a system of type
$(\argument\times A\to 2)\to [0,1]$. This is similar to 
our Theorem~\ref{cor:semi-strong}, which establishes an analogous
transformation (to a system of type $(\argument\times A\to
[0,\infty])\to[0,\infty]$) by a rather more high-level argument.




\smallskip\noindent\textbf{Future work.} We plan to investigate to what extent
our treatment extends to coalgebras $X \to T(X + FX)$ for a
monad $\BBT$ (the branching type) and a functor $F$ (the transition
type) and are interested in both a logical and an equational
characterisation of $B$-bisimulation, and in algorithms to
compute $B$-bisimilarity. 

\bibliographystyle{abbrv}
\bibliography{monads,all2,delta2}

\clearpage\appendix
\allowdisplaybreaks

\section{Appendix: Omitted Proof Details}
\renewcommand{\thesubsection}{A.\arabic{subsection}.}

\subsection*{Proof of Lemma~\ref{lem:emb}}
We need to complete the commuting diagram consisting of the solid arrows below
\[\xymatrix{
  X\ar[rr]^{h} \ar[d]_{f} & & Y \ar[dd]^{g'}
	\ar@{-->}[dl]_{g} \\
	FX \ar[d]_{\kappa_X} \ar[r]^{Fh} & FY \ar[dr]^{\kappa_Y} & \\
	GX \ar[rr]_{Gh} & & GY
}\]
to a commutative diagram including the dashed arrow $g:  Y \to
SY$. If $y = h(x)$ for some $x\in X$ we put $g(y) = (F h
\circ f)(x)$ and we choose $g(y)$ arbitrarily, otherwise.
Note that $g$ is well defined, for $h(x) = h(x')$ implies
$(\kappa_Y \circ Fh \circ f)(x) = (\kappa_Y \circ Fh \circ
f)(x')$ and $(Fh \circ f)(x) = (Fh \circ f)(x')$ follows as
$\kappa$ is injective. It is evident that $g$ makes the above
diagram commute. 

\subsection*{$\bbC$ is a Monad (Section \ref{sec:examples})} 

Recall the definition of $\CCM$ from Section \ref{sec:examples}. We have by definition (cf.~\cite{Jacobs:2008:CTS})
\begin{displaymath}
\CCM X = \left\{M\subseteq [0,\infty)^X_{\omega}\bigm\vert M\neq\emptyset,\forall\xi_i\in M,~\forall p_i.\sum\limits_{i} p_i\leq 1\impl \sum\limits_{i} p_i\cdot \xi_i\in M\right\}.
\end{displaymath}
where we denote by $[0,\infty)^X_{\omega}$ the space of finite valuations, i.e.\ those functions $f:X\to [0,\infty)$ for which $\{x\mid\xi(x)\neq 0\}$ is finite.
Equivalently, $\CCM X$ consists of convex closures of non-empty subsets of $[0,\infty)^X_{\omega}$ containing the trivial valuation identically equal to $0$.
Since our definition slightly deviates from the one in~\cite{Jacobs:2008:CTS}, we check that the result is indeed a monad.

The monad structure on $\CCM$ can be conveniently presented by
regarding morphisms $X\to \CCM Y$ as relations over $X\times
[0,\infty)^Y_{\omega}$. For any $ f:X\to \CCM Y$, let us denote by
$f^{\circ}\subseteq X\times [0,\infty)^Y_{\omega}$ the corresponding
relation. Given $f:X\to \CCM Y$, let $\eta:X\to \CCM X$ and
$f^{\klstar}:\CCM X\to \CCM Y$ be such that
\begin{align*}
\eta^{\circ}(x,\xi) &\text{~~~iff~~~} \xi = \delta_x,\\
(f^\klstar)^{\circ}(S,\xi) &\text{~~~iff~~~} \exists\zeta:X\to[0,\infty)\in S.~\xi\in\left\{\sum_{x\in X} \zeta(x)\cdot\theta_x\mathop{\bigl\vert} \forall x.~f^\circ(x,\theta_x)\right\}  
\end{align*}
where $\delta_x(y)=1$ if $x=y$ and $\delta_x(y)=0$ otherwise; scalar
multiplication and summation is extended to valuations pointwise. We
verify the monad laws (see Section \ref{sec:prelims}).

\begin{itemize}
 \item{}[$\eta^{\klstar}=\id$]: 
\begin{align*}
(\eta^\klstar)^\circ(S,\xi) \text{~~~iff~~~}&\exists\zeta:X\to[0,\infty)\in S.~\xi\in\left\{\sum_{x\in X} \zeta(x)\cdot\delta_x\right\}\\ 
\text{~~~iff~~~}&\exists\zeta:X\to[0,\infty)\in S.~\xi=\zeta\\[1ex]
\text{~~~iff~~~}&\xi\in S.
\end{align*} 
 \item{}[$f^{\klstar}\eta=f$]:
\begin{align*}
(f^{\klstar}\eta)^\circ(x,\xi) \text{~~~iff~~~}&\exists\zeta=\delta_x.~\xi\in\left\{\sum_{x\in X} \zeta(x)\cdot\theta_x\mathop{\bigl\vert} \forall x.~f^\circ(x,\theta_x)\right\}\\
\text{~~~iff~~~}&\xi\in\left\{\theta_x\mid f^\circ(x,\theta_x)\right\}\\[1ex]
\text{~~~iff~~~}&f^\circ(x,\xi). 
\end{align*}
 \item{}[$(f^{\klstar} g)^{\klstar}=f^{\klstar}g^{\klstar}$]: On the one hand,
\begin{align*}
&\bigl((f^{\klstar} g)^{\klstar}\bigr)^{\circ}(S,\xi)\\[1ex]
 \text{~~~iff~~~}& \exists\zeta:X\to[0,\infty)\in S.\\
&\xi\in\left\{\sum_{x\in X} \zeta(x)\cdot\theta_x\mathop{\bigl\vert} \forall x.~(f^\klstar g)^\circ(x,\theta_x)\right\}\\
\text{~~~iff~~~}& \exists\zeta:X\to[0,\infty)\in S.\\
&\xi\in\left\{\sum_{x\in X,y\in Y} \zeta(x)\cdot\zeta_x(y)\cdot\theta_{xy}\mathop{\bigl\vert} \forall x,y.~g^{\circ}(x,\zeta_x)\land f^\circ(y,\theta_{xy}) \right\}
\end{align*}
and on the other hand,
\begin{align*}
&\bigl(f^{\klstar}g^{\klstar}\bigr)^{\circ}(S,\xi)\\[1ex]
 \text{~~~iff~~~}& \exists\zeta':Y\to[0,\infty).~(g^{\klstar})^{\circ}(S,\zeta')~\land\\
&\xi\in\left\{\sum_{y\in Y} \zeta'(y)\cdot\theta_y\mathop{\bigl\vert} \forall y.~f^\circ(y,\theta_y)\right\}\\ 
\text{~~~iff~~~}& \exists\zeta:X\to[0,\infty)\in S.\\
&\xi\in\left\{\sum_{y\in Y} \zeta'(y)\cdot\theta_y\mathop{\bigl\vert} \forall y.~f^\circ(y,\theta_y)\right\}~\land\\
&\zeta'\in\left\{\sum_{x\in X} \zeta(x)\cdot\zeta_x\mathop{\bigl\vert} \forall x.~g^\circ(x,\zeta_x)\right\}\\
\text{~~~iff~~~}& \exists\zeta:X\to[0,\infty)\in S.\\
&\xi\in\left\{\sum_{x\in X,y\in Y} \zeta(x)\cdot\zeta_x(y)\cdot\theta_y\mathop{\bigl\vert} \forall x,y.~g^\circ(x,\zeta_x)\land f^\circ(y,\theta_y)\right\}.
\end{align*}
Note that $\bigl(f^{\klstar}g^{\klstar}\bigr)^{\circ}(S,\xi)$ implies $\bigl((f^{\klstar} g)^{\klstar}\bigr)^{\circ}(S,\xi)$, since any family $\{\theta_y\}_{y\in Y}$ gives rise to a family $\{\theta_{xy}\}_{x\in X,y\in Y}$ vacuously depending on $x$. The converse implication is nontrivial and makes use of convexity. Suppose $\bigl((f^{\klstar} g)^{\klstar}\bigr)^{\circ}(S,\xi)$. Then there is $\zeta:X\to[0,\infty)\in S$ such that
\begin{displaymath}
\xi=\sum_{x\in X,y\in Y} \zeta(x)\cdot\zeta_x(y)\cdot\theta_{xy}
\end{displaymath} 
where the families $\{\zeta_x\}_{x\in X}$ and $\{\theta_{xy}\}_{x\in X,y\in Y}$ satisfy the conditions: for all $x$ and $y$, $g^{\circ}(x,\zeta_x)$ and $f^\circ(y,\theta_{xy})$. For any $y\in Y$ let
\begin{displaymath}
\theta_y = \frac{\sum_{x\in X}\zeta(x)\cdot\zeta_x(y)}{\sum_x\zeta(x)\cdot\zeta_x(y)}\cdot\theta_{xy}
\end{displaymath}
if the denominator is nonzero and $\theta_y=0$ otherwise. In both case we have
\begin{displaymath}
\sum_{x\in X}\zeta(x)\cdot\zeta_x(y)\cdot \theta_y = \sum_{x\in X}\zeta(x)\cdot\zeta_x(y)\cdot\theta_{xy}
\end{displaymath}
and therefore 
\begin{displaymath}
\xi=\sum_{x\in X,y\in Y}\zeta(x)\cdot\zeta_x(y)\cdot \theta_y.
\end{displaymath}
Since $\theta_y$ either identical to $0$ or is a convex combination of the $\{\theta_{xy}\}_{x,y}$, $f^\circ(y,\theta_{xy})$ implies $f^\circ(y,\theta_{y})$ and therefore, according to the previous calculations, $\bigl(f^{\klstar}g^{\klstar}\bigr)^{\circ}(S,\xi)$.\qed 
\end{itemize}

\pacman{
\subsection{Branching bisimulation}
The definition of the branching observation pattern relies on a reformulation of branching bisimulation in terms of weak transitions as follows. Recall the definition from~~\cite{GlabbeekWeijland96}.
\begin{defn}
A symmetric relation $E$ is a \emph{branching bisimulation} on an LTS $(X,\to)$ if $x E y$ with $x\xrightarrow{\tau} x'$ implies $x' E y$; and for any $a\in A$ if $x E y$ and $x\xrightarrow{a} x'$ then there exist $y_1,y_2,y'$ such that $y~(\xrightarrow{\tau})^*~y_1\xrightarrow{a}y_2~(\xrightarrow{\tau})^*~y'$ and $x E y_1$, $x' E y_2$, $x' E y'$. Two states of an LTS are \emph{delay bisimilar} if there is a relating them branching bisimulation.
\end{defn}
Our reformulation is as follows.
\begin{prop}
Given an LTS $(X,\to)$, two states are branching bisimilar iff they are strong bisimilar in the LTS $(X,\To)$ where $\To$ is defined as follows: 
$x \stackrel{a}\To y$ iff $x \stackrel{a}\to y$ and $x \stackrel{\tau}\To y$ iff $x~(\stackrel{\tau}\to)^*~y$.
\end{prop}
\begin{proof}
Suppose $E$ is a branching bisimulation such that $x E y$. If $x\stackrel{\tau}\To x'$ then there is a sequence $x_1,\ldots,x_n$ such that $x_1=x$, $x_n=x'$ and for any $i<n$, $x_i\stackrel{\tau}\to x_{i+1}$. If follows by induction that for every $i$, $x_i E y$ and therefore we obtain $x' E y$ and $y\stackrel{\tau}\To y$. If $x\stackrel{a}\To x'$ then $x\stackrel{\tau}\to x'$ and therefore, by definition, there are $y_1,y_2$ and $y'$ such that $y~(\xrightarrow{\tau})^*~y_1\xrightarrow{a}y_2~(\xrightarrow{\tau})^*~y'$ and $x E y_1$, $x' E y_2$, $x' E y'$.
\end{proof}
}

\subsection*{Proof of Lemma \ref{lem:lfp}}
Existence of $f_{h}^B$ follows from Kleene fixpoint theorem as
$f_{h}^B$ is the least fixpoint of the
continuous functional 
\begin{equation}\label{eq:func}
  F(g) = \lambda x. \lambda b. 
	\left. \begin{cases}
	  \eta(\pi(x)) & \epsilon \in B \\ \bot & \mbox{otherwise}
	\end{cases} \right\} \oplus
	(\lambda (y, a). g(y)(b/a))^\dagger \circ f(x)
\end{equation}
the continuity of which follows from $T$ being completely ordered.
%
%
\subsection*{Proof of Lemma \ref{lem:union}}

We write $E$ for the transitive closure of $\bigcup_{i \in I} E_i$.
Let $(X, f: X \to T(X \times A))$ be given and let $\pi: X \to X
\slash E$ be the canonical projection. To show that $E \subseteq
\Ke(f^B_\pi)$ it suffices to show that $E_i \subseteq
\Ke(f^B_\pi)$ for all $i \in I$. In the sequel, we therefore
fix an arbitrary index $i \in I$. 

We have that $f^B_\pi = \lub_{n \in \omega} f_i$ by Kleene
fixpoint theorem, where each $f_i: X \to B \to T(X \slash R)$ and $f_0(x)(b)
= \bot$, and
\[ f_{i+1}(-)(b) = 
  \left.\begin{cases} \eta(\pi(x)) & \epsilon \in B \\ \bot &
\mbox{otherwise}
  \end{cases}\right\} \oplus (\lambda (y, a). f_i(y)(b / a))^\dagger
\circ f
\]
by definition. We show, by induction, that $E_i \subseteq \Ke(f_i)$
which implies the result. For $i = 0$ there is nothing to show. For
$i > 0$ we fix $(x, x') \in E$. By induction hypothesis, $f_i(x) =
f_i(y)$ and we are done as $\pi(x) = \pi(y)$. 

\subsection*{Proof of Lemma \ref{lem:semiring-enriched}}

The Kleisli category of $\BBT_R$ induced by a continuous semiring $R$
can be equivalently viewed as the category of $R$-valued relations
$X\times Y\to R$, for any such relation is isomorphic to a function
$X\to T_R Y$ and vice versa. The pointed dcpo-structure over Kleisli hom-sets
is then inherited from $R$. The unit of the monad gives rise to the
diagonal relation $\delta:X\times X\to R$ sending $\brks{x,y}$ to
$1$ if $x=y$ and to $0$ otherwise. A composition of two relations
$r:X\times Y\to R$ and $r:Y\times Z\to R$ induced by the Kleisli
composition of $\BBT_R$ is as follows:
\begin{displaymath}
(r\cdot s)(x,z) = \sum\limits_{y\in Y} r(x,y)\cdot r(y,z)
\end{displaymath} 
Now it is clear that continuity of least upper bounds over Kleisli composition (in both arguments) is a direct implication of continuity of multiplication in $R$.
\subsection*{Proof of Lemma \ref{lemma:fully-probabilistic}}

We recapitulate the construction of total probabilities given in
\cite{BaierHermanns97} before giving the proof.
For $x_0 \in X$ we define a measure $\mu(x_0)$ on the
boolean
algebra
$B = \lbrace S \cdot (X \times A)^\omega \mid n \geq 0, S
\subseteq (X
\times A)^n \rbrace$
by putting 
\[ \mu(x_0)(S) = \sum \lbrace f(x_0)(x_1, a_1) \cdot \dots
\cdot f(x_{n-1})(x_n, a_n) \mid (x_1, a_1, \dots, x_n, a_n) \in S
\rbrace \]
for $S \subseteq (A \times X)^n$.
By the Hahn-Kolmogorov theorem, every $\mu(x_0)$  extends to a
measure on the $\sigma$-algebra generated by $B$.  Total
probabilities are now given by
\[ P(x_0, \Lambda, C) = \mu(x_0)\Bigl(\bigcup_{n \geq 0} \lbrace (x_1,
a_1, \dots) \in (X \times A)^\omega\! \mid\! (a_1, \dots, a_n) \in
\Lambda, x_n \in C \rbrace\Bigr)\]
where measurability of the argument of $\mu(x_0)$ is clear, and 
$\mu(x_0)$ measures the probabilities of the system evolving along a
set of paths, starting from $x_0$.
$\mu(x_0)$ 

We can now give the proof of Lemma \ref{lemma:fully-probabilistic}
as follows.
Let $f: X \to (X \times A) \to [0, 1]$ be a fully probabilistic
system, i.e.\ $\sum_{(y, a) \in X \times A} f(x)(y, a) = 1$ for all
$x \in X$. 
If $x_0 \in X, \Lambda \subseteq A^*$ and $Y \subseteq X$ we write
\[ P(x_0)(T, Y) = \lbrace (x_1, a_1, \dots, x_n, a_n) \in (X \times
A)^* \mid (a_1, \dots, a_n) \in T, x_n \in Y \rbrace
\]
for the set of paths that connect $x_0$ to an element in $Y$ via a
trace in $T$, and
\[
  P^- = \lbrace \pi \in P \mid \mbox{ no prefix of $\pi$ is in $P$}
  \rbrace
\]
for the set of minimal prefixes of a set $P \subseteq (X \times
A)^*$. One then verifies that
\[ \mu(x_0)(\Lambda, Y) = \sum\limits_{\substack{(x_1, \dots, a_n) \in
P(x_0)(\Lambda, Y)^-\\n \geq 0}} f(x_0)(x_1, a_1) \cdots f(x_{n-1})(x_n, a_n) \]
as $f(x)(y, a) \in [0, 1]$ for all $x, y \in X$ and all $a \in A$.
Given an equivalence relation $E \subseteq X \times X$ with
associated projection $\pi: X \to X \slash E$, we have that
$f^B_\pi = \bigsqcup_{n \in \omega} f_i$ where each $f_i\colon
X \to B \to T_{[0, \infty]}(X \slash R)$ and $f_0(x)(b)([x']_E) = 0$
and $f_i(x)(b)([x']_E) = 1$ if $\epsilon \in b$ and $(x, x') \in E$,
and
\[ f_i(x)(b)([x']_E) = \sum\limits_{(y, a) \in X \times A} f(x)(y, a) \cdot
f_{i-1}(y)(b/a)([x']_E) 
\] 
otherwise, by applying Kleene's fixpoint theorem and unravelling
Kleisli composition induced by the monad
$T_{[0, \infty]}$. We show that
\begin{equation}\label{eqn:recur}
 f_{i+1}(x)(b)([x']_E) = \sum\limits_{(x_1, \dots,  a_i)
\in P(x)(b, [x']_E)^-}
  f(x)(x_1, a_1) \cdots f(x_{i-1})(x_i, a_i) 
\end{equation}
which entails the claim.  For $i = 0$ there is nothing to show. If
$i > 0$, we have that
\[ f_{i+1}(x)(b)([x']_E) = 1 = \sum\limits_{(x_1, \dots, a_i) \in
P(x)(b, [x']_E)^-} f(x)(x_1, a_1) \cdots f(x_{i-1})(x_i, a_i) \]
if $\epsilon \in b$ and $(x, x') \in E$ by definition. Not suppose
that $x \notin [x']_E$ or $\epsilon \notin b$.

If $S$ is a set, $s \in S$ and $\Lambda \subseteq  S^*$ we write
$\Lambda / s = \lbrace w \in S^* \mid sw \in \Lambda$ for the
Brozowski derivative of $\Lambda$ with respect to $s$. One checks
that $P(x)(\Lambda / a, Y)^- = (P(x)(\Lambda, Y)^-) / (a, y)$ for all
$x, y \in X$ and all $a \in A$, which allows us to verify Equation
(\ref{eqn:recur}) by calculation.

\subsection*{The monad $\CCM$ is completely ordered} 

We use the same notation as earlier and denote by $f^{\circ}\subseteq X\times [0,\infty)^Y_{\omega}$ the relation representing $f:X\to TY$. The Kleisli hom-sets $\Hom(X,TY)$ therefore form a dcpo under the partial order induced by the corresponding order over the relations. Note that, in contrast to~\cite{Jacobs:2008:CTS}, each Kleisli hom-set has a bottom element, represented by the relation $\{\brks{x,\xi}\mid\forall y.\,\xi(y)=0\}$.

We show continuity of Kleisli composition. Given $f:X\times Y\to R$ and $g:Y\times Z\to R$,

\begin{displaymath}
(g^\klstar f)(x,\xi) \text{~~~iff~~~}  \exists\zeta:X\to[0,\infty)\in f(x).~\xi\in\left\{\sum\limits_{y\in Y} \zeta(y)\cdot\theta_y\mathop{\bigl\vert} \forall y.~g^\circ(y,\theta_y)\right\}. 
\end{displaymath}
Now, if $f$ is a least upper bound of a directed set $\{f_i\}_i$ then
\begin{align*}
(g^\klstar f)(x,\xi) \text{~~~iff~~~}  &\exists\zeta:X\to[0,\infty)\in f(x).~\xi\in\left\{\sum\limits_{y\in Y} \zeta(y)\cdot\theta_y\mathop{\bigl\vert} \forall y.~g^\circ(y,\theta_y)\right\}\\
\text{~~~iff~~~}  &\exists i.\,\exists\zeta:X\to[0,\infty)\in f_i(x).~\xi\in\left\{\sum\limits_{y\in Y} \zeta(y)\cdot\theta_y\mathop{\bigl\vert} \forall y.~g^\circ(y,\theta_y)\right\}\\ 
\text{~~~iff~~~}  & \exists i.\,(g^\klstar f_i)(x,\xi)\\[2ex]
\text{~~~iff~~~}  & \brks{x,\xi}\in\bigcup_i g^\klstar f_i.
\end{align*}
Now suppose that $g$ is a least upper bound of a directed set $\{g_i\}_i$. Note that if for some $y\in Y$ and $\theta_y\in TZ$, $g^\circ(y,\theta_y)$ then there is $i$ such that $g_j^\circ(y,\theta_y)$ for any $j\geq i$. Therefore,
\begin{align*}
(g^\klstar f)(x,\xi) \text{~~~iff~~~}  &\exists\zeta:X\to[0,\infty)\in f(x).~\xi\in\left\{\sum\limits_{y\in Y} \zeta(y)\cdot\theta_y\mathop{\bigl\vert} \forall y.~g^\circ(y,\theta_y)\right\}\\
\text{~~~iff~~~}  &\exists\zeta:X\to[0,\infty)\in f(x).~\xi\in\left\{\sum\limits_{y\in Y} \zeta(y)\cdot\theta_y\mathop{\bigl\vert} \forall y.\,\exists i.\,g_i^\circ(y,\theta_y)\right\}\\
\text{~~~iff~~~}  &\exists\zeta:X\to[0,\infty)\in f(x).\,\exists k.~\xi\in\left\{\sum\limits_{y\in Y} \zeta(y)\cdot\theta_y\mathop{\bigl\vert} \forall y.\,g_k^\circ(y,\theta_y)\right\}\\
\text{~~~iff~~~}  &\brks{x,\xi}\in\bigcup_k g_k^\klstar f.
\end{align*}
Here, we made use of the fact that if for all $y\in Y$ there is $i$ such that $g_i^\circ(y,\theta_y)$ then there is $k$ such that forall $y\in Y$, $g_k^\circ(y,\vartheta_y)$ with some family $\{\vartheta_i\}_i$ so that%
\begin{displaymath}
\sum\limits_{y\in Y} \zeta(y)\cdot\theta_y = \sum\limits_{y\in Y} \zeta(y)\cdot\vartheta_y.  
\end{displaymath}
As such $k$ we can take the maximum $\max\{i\mid \forall y\in\mathrm{supp}(\zeta).~g_i^\circ(y,\theta_y)\}$ which exists because $\mathrm{supp}(\zeta)$ is finite and then put $\vartheta_i=\theta_i$ if $i\leq k$ and $\vartheta_i=\top$ otherwise.\qed

\section*{Modelling Simple Segala Systems with $\bbC$}
According to~\cite{Segala:1994:PSP}, a simple Segala system is a coalgebra $X\to\PSet(\Dist X\times A)$ where $\Dist$ refers to finite distribution functor, i.e.\ $\Dist X$ consists of the valuations $\xi:X\to [0,\infty)$ satisfying two conditions:
\begin{itemize}
 \item $\textrm{supp}(\xi)=\{x\mid \xi(x)\neq 0\}$ is finite;
 \item $\sum_{x\in X}\xi(x)=1$.
\end{itemize}
In fact the original definition in~\cite{Segala:1994:PSP} is formulated in terms of probability spaces and does not involve any cardinality restrictions. However, restricting to finite or countable distributions is a common practice. For our sakes we restrict to finite distributions. 

 In order to match the presentation from~\cite{Segala:1994:PSP} we use the notation $x\xrightarrow{a}\xi$ iff $\brks{\xi,a}\in f(x)$ where $(X,f:X\to\PSet(\Dist X\times A))$ is some fixed simple Segala system. In these terms, recall from~\cite{Segala:1994:PSP} that a \emph{combined step} $x\stackrel{a}{\rightsquigarrow}\xi$ encodes the following:
\begin{displaymath}
\exists\xi_i.~\forall i.~ x\xrightarrow{a}\xi_i\land\exists r_i.~\sum_i r_i=1\land \xi=\sum_i r_i\cdot\xi_i.
\end{displaymath}
Informally, a combined step is a convex combination of ordinary steps.
\begin{defn}[Strong probabilistic bisimulation~\cite{Segala:1994:PSP}]
An equivalence relation $E\subseteq X\times X$ on a simple Segala system $(X,f:X\to\PSet(\Dist X\times A))$ is a \emph{strong probabilistic bisimulation} iff for any $x,y\in X$ such that $x E y$ and $x\xrightarrow{a}\xi$ there is a combined step $y\stackrel{a}{\rightsquigarrow}\xi'$ such that for any $E$-equivalence class $C$, $\sum_{z\in C}\xi(z)=\sum_{z\in C}\xi'(z)$.
\end{defn}
We provide a translation of simple Segala systems to $\bbC(\argument\times A)$-coalgebras by postcompoing the coalgebra map with the following natural transformation $\kappa_X:\PSet(\Dist X\times A))\to\CCM(X\times A)$:
\begin{displaymath}
\kappa_X(S\in \PSet(\Dist X\times A))) = \left\{\lambda\brks{x,a}.~\sum_i r_i\cdot\delta_{a,a_i}\cdot\xi_i(x)\bigm\vert \sum_i r_i\leq 1,\brks{\xi_i,a_i}\in S \right\}
\end{displaymath}  
where $\delta_{a,a}=1$ and $\delta_{a,b}=0$ if $a\neq b$. Note that $\kappa$ is not injective. Yet, coalgebraic bisimulations over the translated system precisely capture probabilistic strong bisimulations of the original ones.
\begin{lem}\label{lem:sseg}
An equivalence relation $E\subseteq X\times X$ on a simple Segala system $(X,f:X\to\PSet(\Dist X\times A))$ is a strong probabilistic bisimulation iff it is a kernel bisimulation on $(X,\kappa_X f)$. 
\end{lem} 
\begin{proof}
Let us fix a simple Segala system $(X,f:X\to\PSet(\Dist X\times A))$. The claim that $E$ is a kernel bisimulation $E$ on $(X,\kappa_X f)$ can be spelled as follows: for any $x,y\in X$ if $x E y$ then $((F\pi)\kappa_X f)(x)=((F\pi)\kappa_X f)(y)$ where $\pi:X\to X/E$ is the canonical projection and $F=\CCM(\argument\times A)$. Due to the presupposed symmetry of $E$, $((F\pi)\kappa_X f)(x)=((F\pi)\kappa_X f)(y)$ is equivalent to the inclusion $((F\pi)\kappa_X f)(x)\subseteq ((F\pi)\kappa_X f)(y)$. Note that we have
\begin{align*}
((F\pi)\kappa_X f)(x)=\,&(F\pi)\left\{\lambda\brks{z,a}.\sum_i r_i\cdot\delta_{a,a_i}\cdot\xi_i(z)\bigm\vert x\xrightarrow{a_i}\xi_i, \sum_i r_i\leq 1\right\}\\
=\,&\left\{\lambda\brks{C,a}.\sum_{z\in C}\sum_i r_i\cdot\delta_{a,a_i}\cdot\xi_i(z)\bigm\vert x\xrightarrow{a_i}\xi_i, \sum_i r_i\leq 1\right\}.
\end{align*}  
Therefore, the inclusion $((F\pi)\kappa_X f)(x)\subseteq((F\pi)\kappa_X f)(y)$ amounts to the following: if $x\xrightarrow{a_i}\xi_i$ and $\sum_i r_i\leq 1$ then there are $\zeta_j$, $b_j$ and $s_j$ such that $y\xrightarrow{b_i}\zeta_j$, $\sum_j s_j\leq 1$ and 
\begin{align}\label{eq:pbis}
\lambda\brks{C,a}.\sum_{z\in C}\sum_i r_i\cdot\delta_{a,a_i}\cdot\xi_i(z)=\lambda\brks{C,b}.\sum_{z\in C}\sum_j s_j\cdot\delta_{b,b_j}\cdot\zeta_j(z).
\end{align}
In particular, if the family $\{r_i\}_i$ is the singleton $\{1\}$ then $\sum_{z\in C}\xi_1(z)=\sum_{z\in C}\sum_{b_j=a_1} s_j\cdot\zeta_j(z)$. Further summation over equivalence classes $C$, yields $1=\sum_{z}\sum_{b_j=a_1} s_j\cdot\zeta_j(z)= \sum_{b_j=a_1} s_j$, which implies in particular that $y\stackrel{a_1}{\rightsquigarrow}\sum_{b_j=a_1} s_j\cdot\zeta_j$. In summary we obtained that $E$ must be a probabilistic strong bisimulation on $(X,f)$. 

In order to complete the proof we need to show that the remaining conditions with $\{r_i\}_i$ not being $\{1\}$ are derivable. By the above reasoning we can assume that for any $i$, there are is a family of nonnegative reals $\{t_{ij}\}_j$  and a family of distributions $\{\zeta_{ij}\}_j$ such that $\sum_{z\in C}\xi_i(z)=\sum_{z\in C}\sum_{j} t_{ij}\cdot\zeta_{ij}(z)$, $\sum_j t_{ij}=1$ and for every $j$, $y\xrightarrow{a_i}\zeta_{ij}$. Therefore
\begin{align*}
\lambda\brks{C,a}.\sum_{z\in C}\sum_i r_i\cdot\delta_{a,a_i}\cdot\xi_i(z)=\lambda\brks{C,b}.\sum_{z\in C}\sum_{i,j} r_i\cdot t_{ij}\cdot\delta_{b,a_i}\cdot\zeta_{ij}(z).
\end{align*}
Since $\sum_{i,j} r_i\cdot t_{ij}=\sum_{i} r_i\leq 1$, we have thus indeed obtained an instance of equation~\eqref{eq:pbis}.
\qed\end{proof}
We now turn our attention to probabilistic weak bisimulation for simple Segala systems, again going back to~\cite{Segala:1994:PSP}.
\begin{defn}[Weak probabilistic bisimulation~\cite{Segala:1994:PSP}]
An equivalence relation $E\subseteq X\times X$ on a simple Segala system $(X,f:X\to\PSet(\Dist X\times A))$ is a \emph{weak probabilistic bisimulation} iff for any $x,y\in X$ such that $x E y$ and $x\xrightarrow{a}\xi$ there is $n$ such that $y\xRightarrow{a}_n\xi'$ and for any $E$-equivalence class $C$, $\sum_{z\in C}\xi(z)=\sum_{z\in C}\xi'(z)$. The family of relations $\xRightarrow{a}_n$ is defined by induction as follows:
\begin{itemize}
 \item $x\xRightarrow{a}_0\xi$ iff $a=\tau$ and $\xi=\delta_x$;
 \item $x\xRightarrow{a}_{n+1}\xi$ iff $\xi$ is representable as 
$
\sum_{y\in X}p\cdot\zeta_a(y)\cdot\theta^\tau_y + (1-p)\cdot\zeta_{\tau}(y)\cdot\theta^{a}_y
$
where $p\in[0,1]$, $x\stackrel{a}{\rightsquigarrow}\zeta_a$, $x\stackrel{\tau}{\rightsquigarrow}\zeta_{\tau}$ and for any $y\in Y$, $y\xRightarrow{a}_n\theta_y^a$, $y\xRightarrow{\tau}_n\theta_y^\tau$ (note that if $a=\tau$ this requires $\xi$ to be equal $\sum_{y\in X}\zeta_\tau(y)\cdot\theta^\tau_y$). 
\end{itemize}
\end{defn}
For any $a\in A$, let us denote by $\xRightarrow{a}$ the union of all relations $\xRightarrow{a}_n$. Note that the resulting relation determines a new simple Segala system $(X,g:X\to\PSet(\Dist X\times A))$ for which $g(x)(\xi,a)$ iff $x\xRightarrow{a}\xi$.
\begin{lem}\label{lem:seqw}
Let $(X,f:X\to\PSet(\Dist X\times A))$ be a simple Segala system inducing a family of relations $\xRightarrow{a}$. 
\begin{itemize}
 \item Suppose, $x\xRightarrow{a}\xi$ and for any $y\in X$, $y\xRightarrow{\tau}\xi_y$. Then $x\xRightarrow{a}\sum_{y\in X}\xi(y)\cdot\xi_y$. 
 \item Suppose, $x\xRightarrow{\tau}\xi$ and for any $y\in X$, $y\xRightarrow{a}\xi_y$. Then $x\xRightarrow{a}\sum_{y\in X}\xi(y)\cdot\xi_y$. 
\end{itemize}
\end{lem}
\begin{proof}
We prove only the first clause, as the second one is completely analogous. Suppose, $x\xRightarrow{a}_n\xi$ and proceed by induction over $n$. 

Let $n=0$. Then, by definition, $a=\tau$ and $\xi=\delta_x$. Therefore, $x\xRightarrow{a}\sum_{y\in X}\delta_x(y)\cdot\xi_y=\xi_x$.

Let $n>0$. Then $\xi=\sum_{y\in X}p\cdot\zeta_a(y)\cdot\theta^\tau_y + (1-p)\cdot\zeta_{\tau}(y)\cdot\theta^{a}_y$ where $p\in[0,1]$, $x\stackrel{a}{\rightsquigarrow}\zeta_a$, $x\stackrel{\tau}{\rightsquigarrow}\zeta_{\tau}$ and for any $y\in Y$, $y\xRightarrow{a}_{n-1}\theta_y^a$, $y\xRightarrow{\tau}_{n-1}\theta_y^\tau$. By the induction hypothesis, for all $y\in X$, $y\xRightarrow{a}\sum_{z\in X}\theta_y^a(z)\cdot\xi_z$, $y\xRightarrow{\tau}\sum_{z\in X}\theta_y^\tau(z)\cdot\xi_z$. Therefore
\begin{align*}
x\xRightarrow{a}_n~&\sum_{y\in X}\left(p\cdot\zeta_a(y)\cdot\sum_{z\in X}\theta_y^\tau(z)\cdot\xi_z+
(1-p)\cdot\zeta_\tau(y)\cdot\sum_{z\in X}\theta_y^a(z)\cdot\xi_z\right)\\
=~~~&\sum_{z\in X}\left(\sum_{y\in X} p\cdot\zeta_a(y)\cdot\theta_y^\tau(z)+
(1-p)\cdot\zeta_\tau(y)\cdot\theta_y^a(z)\right)\cdot\xi_z
\end{align*}
and we are done.
\qed\end{proof}
\begin{lem}\label{lem:pseg}
An equivalence relation $E\subseteq X\times X$ on a simple Segala system $(X,f:X\to\PSet(\Dist X\times A))$ is a \emph{weak probabilistic bisimulation} iff the induced family of relations $\xRightarrow{a}$ determines the simple Segala system on which $E$ is a strong probabilistic bisimulation. 
\end{lem} 
\begin{proof}
Let $(X,g:X\to\PSet(\Dist X\times A))$ be the simple Segala system corresponding to the relation $\xRightarrow{a}$. Suppose, $E$ is a strong probabilistic bisimulation on $(X,g)$ and let us show that $E$ is a weak probabilistic bisimulation on $(X,f)$. Suppose that $x E y$ and $x\xrightarrow{a}\xi$. Clearly, $x\xRightarrow{a}_1\xi$ and therefore $x\xRightarrow{a}\xi$. By definition, there is a family $\{\xi_i\}_i$ of distributions and a family of nonnegative reals $\{r_i\}_i$ such that $\sum_i r_i=1$ and for all $i$, $y\xRightarrow{a}\xi_i$ and any $E$-equivalence class~$C$, $\sum_{z\in C}\xi(z)=\sum_{z\in C}\sum_i r_i\cdot\xi_i(z)$. By definition, for any $i$ there is $n_i$ such that $y\xRightarrow{a}_{n_i}\xi_i$ and therefore for all $i$, $y\xRightarrow{a}_n\xi_i$ where $n=\max_i n_i$. It is then easy to check by induction over $n$ that $y\xRightarrow{a}_n\sum_i r_i\cdot\xi_i$. In summary we obtain that $\sum_{z\in C}\xi(z)=\sum_{z\in C}\xi'_i(z)$ for any $E$-equivalence class~$E$ with $\xi'=\sum_i r_i\cdot\xi_i$.

Now suppose that $E$ is a weak probabilistic bisimulation on $(X,f)$ and show that $E$ is a strong probabilistic bisimulation on $(X,g)$. Suppose that $x E y$ and $x\xRightarrow{a}_n\xi$. It is then sufficient to construct such $\xi'$ that $y\xRightarrow{a}\xi'$ and for any $E$-equivalence class~$C$, $\sum_{z\in C}\xi(z)=\sum_{z\in C}\xi'(z)$, which we do by induction over~$n$.

If $n=0$ then $a=\tau$ and $\xi=\delta_x$. Then we are done with $\xi'=\delta_y$.

If $n>0$ then $\xi$ must have form
$
\sum_{z\in X}p\cdot\zeta_a(z)\cdot\theta^\tau_z + (1-p)\cdot\zeta_{\tau}(z)\cdot\theta^{a}_z
$
where $p\in[0,1]$, $x\stackrel{a}{\rightsquigarrow}\zeta_a$, $x\stackrel{\tau}{\rightsquigarrow}\zeta_{\tau}$, for all $z\in X$, $z\xRightarrow{a}_{n-1}\theta_z^a$, $z\xRightarrow{\tau}_{n-1}\theta_z^\tau$. Using the fact that $E$ is a weak probabilistic bisimulation it is straightforward to construct such $\zeta_a'$ and $\zeta_{\tau}'$ that $y\xRightarrow{a}\zeta_a'$, $y\xRightarrow{\tau}\zeta_\tau'$ and for any $E$-equivalence class~$C$, $\sum_{z\in C}\zeta_a(z)=\sum_{z\in C}\zeta_a'(z)$ and $\sum_{z\in C}\zeta_{\tau}(z)=\sum_{z\in C}\zeta_{\tau}'(z)$. By induction, for any $z\in X$ and any $b\in\{a,\tau\}$ there exists $\vartheta_z^b$ such that $z\xRightarrow{b}\vartheta_{z}^b$ and for any $E$-equivalence class~$C$, $\sum_{v\in C}\vartheta_{z}^b(v)=\sum_{v\in C}\theta_{z'}^b(v)$ whenever $z E z'$. We then define $\xi'=\sum_{z\in X} p\cdot\zeta'_a(z)\cdot\vartheta_z^\tau+(1-p)\cdot\zeta'_\tau(z)\cdot\vartheta_z^a$.

By Lemma~\ref{lem:seqw}, $y\xRightarrow{a}\sum_{z\in X}\zeta'_a(z)\cdot\vartheta_z^\tau$ and $y\xRightarrow{a}\sum_{z\in X}\zeta'_\tau(z)\cdot\vartheta_z^a$. Therefore, it is easy to see that $y\xRightarrow{a}\xi'$. For any $C\in X/E$ we have
\begin{align*}
\sum_{v\in C}\xi(v) =~&\sum_{v\in C}\sum_{z\in X}(p\cdot\zeta_a(z)\cdot\theta^\tau_z(v) + (1-p)\cdot\zeta_{\tau}(z)\cdot\theta^{a}_z(v))\\
=~&p\cdot\sum_{z\in X}\zeta_a(z)\cdot\sum_{v\in C}\theta^\tau_z(v)~+\\
 &(1-p)\cdot\sum_{z\in X}\zeta_{\tau}(z)\cdot\sum_{v\in C}\theta^{a}_z(v)\\
=~&p\cdot\sum_{D\in X/E}\sum_{z\in D}\zeta_a(z)\cdot\sum_{v\in C}\theta^\tau_z(v)~+\\
 &(1-p)\cdot\sum_{D\in X/E}\sum_{z\in D}\zeta_{\tau}(z)\cdot\sum_{v\in C}\theta^{a}_z(v)\\
=~&p\cdot\sum_{D\in X/E}\sum_{z\in D}\zeta_a'(z)\cdot\sum_{v\in C}\vartheta^\tau_z(v)~+\\
 &(1-p)\cdot\sum_{D\in X/E}\sum_{z\in D}\zeta_{\tau}'(z)\cdot\sum_{v\in C}\vartheta^{a}_z(v)\\
=~&\sum_{v\in C}\xi'(v).
\end{align*}
Therefore the induction is finished.
\qed\end{proof}
Finally, we show that our notion of weak bisimulation given by~\eqref{eq:rec} agrees with the notion of weak probabilistic bisimulation for simple Segala systems.
\begin{thm}
An equivalence relation $E$ is a probabilistic bisimulation over a simple Segala system $(X,f:X\to\PSet(\Dist X\times A))$ iff $E$ is a $B$-bisimulation over $(X,\kappa_X f)$ w.r.t.\ to the weak observation pattern. 
\end{thm}
\begin{proof}
Let the family of relations $\xRightarrow{a}$ be induced by $(X,f)$. By Lemma~\ref{lem:pseg}, $E\subseteq X\times X$ is a weak probabilistic bisimulation over $(X,f)$ iff $E$ is a strong probabilistic bisimulation over $(X,g:X\to\PSet(\Dist X\times A))$ where $\brks{\xi,a}\in g(x)$ iff $x\xRightarrow{a}\xi$. By Lemma~\ref{lem:sseg} the latter means that $E$ is a kernel bisimulation on $(X,\kappa_Xg)$, i.e.\ for all $x,y\in X$, if $x E y$ and $\xi\in(\kappa_X g)(x)$ then there is $\xi'$ such that $\xi'\in(\kappa_X g)(x)$ and for all $C\in X/E$ and $a\in A$, $\sum_{z\in C}\xi(z,a)=\sum_{z\in C}\xi'(z,a)$. By definition, $\xi\in(\kappa_X g)(x)$ iff $\xi$ is of the form $\lambda\brks{x,a}.\sum_i r_i\cdot\delta_{a,a_i}\cdot\xi_i(x)$ where $\sum_i r_i\leq 1$ and for all $i$, $x\xRightarrow{a_i}\xi_i$. After simple calculations we conclude that $E$ is a probabilistic bisimulation over $(X,f)$ iff whenever $x E y$ and $x\xRightarrow{a}\xi$ then there is $\xi'$ such that $y\xRightarrow{a}\xi'$ and $\sum_{z\in C}\xi(z)=\sum_{z\in C}\xi'(z)$ for any $E$-equivalence class~$C$.

On the other hand, $E$ is a $B$-bisimulation iff whenever $x E y$ and $x\xRightarrow{\hat a}
\xi$, there exists $\xi'$ such that $y\xRightarrow{\hat a} \xi'$ and $\sum_{z\in C}\xi(z)=\sum_{z\in C}\xi'(z)$ for any $E$-equivalence class~$C$ where $\xRightarrow{\hat a}\,=\bigcup_i\xRightarrow{\hat a}_i$ and $\xRightarrow{\hat a}_i$ is given recurrently as follows:
\begin{align*}
x\xRightarrow{\hat a}_0\zeta&\text{~~iff~~} a=\tau \text{~~and~~} \zeta=\delta_{x}\\
x\xRightarrow{\hat a}_{n+1}\zeta&\text{~~iff~~} \exists\xi\in(\kappa_Xf)(x).~\zeta\in\left\{\sum\limits_{y\in X}\xi(y,a)\cdot\theta^\tau_{y}+\xi(y,\tau)\cdot\theta^a_y\mathop{\Bigm\vert }\forall y.\, y\xRightarrow{\hat b}_{n}\theta^b_{y}\right\}  
\end{align*}
Recall that $\xi\in(\kappa_Xf)(x)$ iff $\xi=\lambda\brks{x,a}.\sum_i r_i\cdot\delta_{a,a_i}\cdot\xi_i(x)$ where $\sum_i r_i\leq 1$ and for all $i$, $\brks{\xi_i,a_i}\in f(x)$. Therefore we can rewrite the former definition as follows:
\begin{align*}
x\xRightarrow{\hat a}_0\zeta&\text{~~iff~~} a=\tau \text{~~and~~} \zeta=\delta_{x}\\[1ex]
x\xRightarrow{\hat a}_{n+1}\zeta&\text{~~iff~~} \zeta=\sum_{y\in X}r\cdot\xi_a(y)\cdot\theta^\tau_{y}+s\cdot\xi_\tau(y)\cdot\theta^a_y\\
&\text{~~where~~}\forall y\in X.~\forall b\in \{a,\tau\}.~y\xRightarrow{\hat b}_{n}\theta^b_{y},~~x\stackrel{a}{\rightsquigarrow}\xi_a,~~x\stackrel{\tau}{\rightsquigarrow}\xi_\tau,~~r+s\leq 1 
\end{align*}
where $\stackrel{b}{\rightsquigarrow}\,\in X\times [0,\infty)^X$ is the combined step relation associated with $(X,f)$. It is easy to check by induction that $x\xRightarrow{\hat a}\xi$ iff either $\xi$ is identically $0$ or $x\xRightarrow{a}(1/\sum_z\xi(z))\cdot\xi$, which implies that indeed $E$ is a weak probabilistic bisimulation iff $E$ is a $B$-bisimulation.
\qed\end{proof}

\subsection*{Proof of Lemma~\ref{lem:alg}}
We will use the following fact, saying that any $n$-ary algebraic operation is exactly specified by an element of $Tn$, called \emph{generic effect}.
\begin{lem}[\cite{PlotkinPower02}]\label{lem:gen}
For any $n$-ary algebraic operation $\alpha$ of\/ $\BBT$ there is a generic effect $e\in Tn$ such that $\alpha_X(m)=m^{\klstar}(e)$. This defines a bijective correspondence between $n$-ary algebraic operations and elements of $Tn$.
\end{lem}
\begin{proof}
Given $\alpha:T^n\to T$ we obtain the corresponding generic effect by applying $\alpha_{n}$ to $\eta_n:n\to Tn$. It is easy to verify that this yields the bijection in question. 
\qed\end{proof}
The fact that each $F_i$ is continuous means that for any directed set $D\subseteq\Hom(X,TY)$ and for any $i\in n$, $F_i\left(\bigsqcup_{f\in D}\right)=\bigsqcup_{f\in D} F_i(f)$. By Lemma~\ref{lem:gen}, there is $e\in Tn$ such that $\alpha_{Y'}(m)=m^{\klstar}(e)$. Therefore
\begin{align*}
\alpha_{Y'}\left(\lambda i.\, F_i\left(\bigsqcup_i f_i\right)(x)\right) 
=&\, \left(\lambda i.\, \bigsqcup_{f\in D} F_i(f_i)(x)\right)^{\klstar}(e)\\
=&\, \left(\bigsqcup_{f\in D} \lambda i.\, F_i(f_i)(x)\right)^{\klstar}(e)\\
=&\, \bigsqcup_{f\in D} (\lambda i.\,F_i(f_i)(x))^{\klstar}(e)\\
=&\, \bigsqcup_{f\in D} \alpha_{Y'}(\lambda i.\,F_i(f_i)(x)).
\end{align*}
%
\qed

\subsection*{Proof of Lemma~\ref{lem:rec_tail}}
We will make use of the property of least fixpoints of continuous
functions called \emph{uniformity}~\cite{Gunter92}
and which can be verified by
fixpoint induction: if $F$, $G$ and $U$ are continuous and satisfy
equations $UF=GU$, $U(\bot)=\bot$ then $U\mu F=\mu G$. 

Let $U(g)=T^Bu\comp g$, let $F$ be the map \eqref{eq:func} and let
$G$ be the analogous map defining $f_{\sigma,u\comp h}$. Obviously
$U(\bot)$. Let us verify the equality $UF=GU$, where we write
$\rho_{\epsilon \in b} = \eta_X $ if $\epsilon \in b$, and
$\rho_{\epsilon \in b} = \bot$, otherwise.  
\begin{align*}
(UF(g))(x)(b) =&~T^Bu\comp(\rho_{\eps\in b}(h(x))\oplus g^\klstar(T(\id\times\lambda a.\, b/a) f(x)))\\
=&~Tu\comp\rho_{\eps\in       b}(h(x))\oplus T^Bu\comp g^\klstar(T(\id\times\lambda a.\, b/a) f(x))\\
=&~\rho_{\eps\in b}(u(h(x)))\oplus (T^Bu\comp g)^\klstar(T(\id\times\lambda a.\, b/a) f(x))\\
=&~(GU(g))(x)(b).
\end{align*}
Therefore, $f^B_{u\comp h}=\mu G=U\mu F=T^Bu\comp f^B_{h}$ and we are done.
\qed

\subsection*{Proof of Lemma~\ref{lem:dia}}
We proceed analogously to
the proof of Lemma~\ref{lem:rec_tail}. Specifically, let
$U(w)=w\comp h$ and let $F$ and $G$ be such that $\mu
G=f^B_{u\comp h}$ and $\mu F=g^B_{u}$. Obviously by
definition $U(\bot)=\bot$. Furthermore,
\begin{align*}
(UF(w))(x)(b) =&~\rho_{\eps\in b}(u(h(x)))\oplus w^\klstar(T(\id\times\lambda a.\, b/a) g(h(x)))\\
 =&~\rho_{\eps\in b}(u(h(x)))\oplus w^\klstar(T(\id\times\lambda a.\, b/a) T_A h(f(x)))\\
 =&~\rho_{\eps\in b}(u(h(x)))\oplus (w\comp h)^\klstar(T(\id\times\lambda a.\, b/a) f(x))\\
 =&~(GU(w))(x)(b)
\end{align*} 
from which we conclude $f^B_{u\comp h}=\mu G=U\mu F=f^B_{u}\comp h$ and thus finish the proof.
\qed

\subsection*{Proof of Theorem~\ref{thm:red}}
For one thing, $E$ is a bisimulation for $f^B_{\id}$ if $E$ is the
kernel of a $T^B$-coalgebra morphism $h:X\to Y$. In this case
$g\comp h=T^B h\comp f^B_{\id}$ and by Lemma~\ref{lem:rec_tail}
$g\comp h=f^B_{h}$, which certainly implies $E=\Ke h\subseteq \Ke
f^B_{h}$.

Let us show the converse implication. Suppose, $E$ is a
$B$-bisimulation, i.e.\ $E\subseteq \Ke f^B_{\pi}$. We can now turn
$Y=X/E$ into a $T^B$-coalgebra as follows $g:Y\to T^B Y$: for every
$[x]_E\in Y$ let $g([x]_E)=f^B_{\pi}(x)$ --- this is correct
because, by assumption $f^B_{\pi}$ does not distinguish
$E$-equivalent elements. By construction we have $g\comp\pi=f^B_{\pi}$. The rest of the argument follows from diagram
\begin{align*}
\xymatrix{
  X \ar[r]^{\pi}\ar[rd]^{f^B_{\pi}} \ar[d]_{f^B_{\id}} & X/E \ar[d]^{g} \\
	T^B X \ar[r]_{T^B \pi} & T^B (X/E)
}
\end{align*}
whose commutativity follows from Lemma~\ref{lem:rec_tail}.
\qed

\subsection*{Proof of Lemma~\ref{lem:conalg}}
Observe that the Kleisli category of the continuation monad over $D$ can be equivalently viewed as a category whose objects are sets and whose morphisms from $X\to Y$ are functions $D^Y\to D^X$ under standard composition, but in the opposite direction. The corresponding subcategory consisting only of continuous maps $D^Y\to_c D^X$ can be viewed as a Kleisli category, for $D^Y\to_c D^X$ is isomorphic to $X\to TY$. This induces a monad structure over $T$. By definition, $\BBT$ is completely ordered: Kleisli hom-sets $\Hom_{\Set_{\BBT}}(X,Y)\cong(D^Y\to_c D^X)$, ordered pointwise inherit the pointed dcpo structure from $D$; this structure is preserved by Kleisli composition by definition.

Let us show that any continuous operation of $\BBT$ is algebraic. 
For simplicity we only consider the case $n=2$. Let $\oplus$ be the
operation $T^2\to T$ of interest. In order to show that $\oplus$  is
algebraic it suffices to show that for any $f,g:(X\to D)\to_c D$ and
$h:X\to ((Y\to D)\to_c D)$, $h^{\klstar}(f\oplus
g)=h^{\klstar}(f)\oplus h^{\klstar}(g)$. Indeed, we have
\begin{align*} 
h^{\klstar}(f\oplus g)(c:Y\to D) =&~(f\oplus g)(\lambda x.\,h(x)(c))\\
=&~f(\lambda x.\,h(x)(c))\oplus f(\lambda x.\,h(x)(c))\\ 
=&~h^{\klstar}(f)(c)\oplus h^{\klstar}(g)(c),
\end{align*} 
which completes the proof.  \qed

\subsection*{Proof of Lemma \ref{lem:ext}}
%
Let us denote by $\hat f$ the composition $\kappa\comp f$. We would
like to show that $\hat f^B_{\pi}= \kappa^B\comp f^B_{\pi}$. This
would automatically imply the claim, for by injectivity of $\kappa$
it would mean that $\Ke f^B_{\pi}=\Ke \hat f^B_{\pi}$, and hence, by
definition, $R$ would be a $B$-bisimulation on $(X,f)$ exactly when
it would be a $B$-bisimulation on $(X,\hat f)$.

As in the proof of Lemma~\ref{lem:rec_tail} we use the uniformity principle. Specifically, we conclude $U\mu F=\mu G$ from $U(\bot)=\bot$ and $UF=GU$ where we take $U(g)=\kappa^B\comp g$, $F$ to be the map~\eqref{eq:func}, and $G$ to be of the same form as $F$ but with $f$ replaced with $\kappa\comp f$. 

Since $\kappa$ is a monad morphism, we certainly have $U(\bot)=\bot$. Then we verify $UF=GU$ as follows:
\begin{align*}
(UF(g))(x)(b) =&~\kappa^B\comp(\rho_{\eps\in b}(\pi(x))\oplus g^\klstar(T(\id\times\lambda a.\, b/a) f(x)))\\
=&~\rho_{\eps\in b}(\pi(x))\oplus \kappa^B g^\klstar(T(\id\times\lambda a.\, b/a) f(x))\\
=&~\rho_{\eps\in b}(\pi(x))\oplus (\kappa\comp g)^\klstar(T(\id\times\lambda a.\, b/a) f(x))\\
=&~(GU(g))(x)(b)
\end{align*}
where we made an essential use of the facts that $\kappa$ was a monad morphism coherently preserving $\oplus$.
\qed
%

\subsection*{Proof of Theorem~\ref{cor:semi-strong}}
Theorem~\ref{cor:semi-strong} follows from the following more general result.
\begin{lem}\label{thm:cont}
Let $\BBT$ be a completely ordered monad for which the following
condition holds: given any $X$, there is a jointly monic family of
the form $\{c_i^{\klstar}:TX\to T1\}_i$.
Then $E\subseteq X\times X$ is a $B$-bisimulation on a 
coalgebra $(X,f:X\to T (X \times A))$ iff it is a kernel bisimulation on $(X,
\kappa_X \circ  f:X\to\widehat T(X \times A))$ where $\widehat
TX=(X\to T1)\to_c T1$
and $\kappa_X$ is the natural transformation $\kappa_X:TX\to (X\to
T1)\to T1$ defined as $\kappa_X(p)(c) = c^\klstar(p)$ and $\oplus$
is defined on $\widehat\BBT$ by equation $(f\oplus g)(c)=f(c)\oplus
g(c)$.  \end{lem}
\begin{proof}
The transformation~$\kappa$ is a monad morphism~\cite{Moggi89}. We order $\widehat T$ by putting $f\sqsubseteq g$ for $f,g\in\widehat T X$ iff for all $c:X\to T1$, $f(c)\sqsubseteq g(c)$. Under this definition, $p\sqsubseteq q$ implies $\kappa_X(p)(c)=c^\klstar(p)\sqsubseteq c^\klstar(q)=\kappa_X(q)(c)$ for all $c$ and hence the ordering is preserved by $\kappa$. Continuity of $\kappa$ follows  analogously. Also we define $\bot$ on $\widehat\BBT$ as the constant map $\lambda c.\,\bot$. Clearly, under this definition, $\bot$ is preserved. We have thus shown that $\kappa$ is a completely ordered monad morphism.

By assumption, for any distinct $p,q\in TX$ there exists $c_i$ such that $\kappa(p)(c_i)\neq\kappa(q)(c_i)$, which means that $\kappa$ is componentwise injective. We are done by Lemmas~\ref{lem:ext},~\ref{lem:conalg} and Theorem~\ref{thm:red}.
\qed
\end{proof}
\end{document}
